\documentclass[3p,11pt]{elsarticle}

\usepackage{lineno,hyperref}
\usepackage{amssymb,amsmath}
\usepackage{amsthm}

\newtheorem{lemma}{Lemma}
\newtheorem{proposition}{Proposition}
\newtheorem{definition}{Definition}
\newtheorem{assumption}{Assumption}
\newtheorem{remark}{Remark}
\usepackage{bbm}
\usepackage{caption}
\usepackage{xcolor}
\usepackage{setspace}
\modulolinenumbers[5]

\usepackage{subcaption}
\usepackage{bm}

\usepackage{booktabs}

\usepackage{numcompress}\biboptions{authoryear}

\makeatletter
\def\ps@pprintTitle{%
\let\@oddhead\@empty
\let\@evenhead\@empty
\def\@oddfoot{}%
\let\@evenfoot\@oddfoot}
\makeatother

\begin{document}

\begin{frontmatter}

\title{A Hawkes model with CARMA(p,q) intensity}


\author{Lorenzo Mercuri}
\address{Department of Economics Management and Quantitative Methods, University of Milan, Milan, Italy \\ email: \texttt{lorenzo.mercuri@unimi.it}}

\medskip

\author{Andrea Perchiazzo}
\address{Faculty of Economic and Social Sciences and Solvay Business School, Vrije Universiteit Brussel, Brussels, Belgium\\ email: \texttt{andrea.perchiazzo@vub.be}}

\medskip

\author{Edit Rroji}
\address{Department of Statistics and Quantitative Methods, University of Milano-Bicocca, Milan, Italy \\ email: \texttt{edit.rroji@unimib.it}}

%
%

\begin{abstract}
In this paper we introduce a new model, named CARMA(p,q)-Hawkes, as the Hawkes model with exponential kernel implies a strictly decreasing behaviour of the autocorrelation function while empirical evidences reject its monotonicity. The proposed model is a Hawkes process where the intensity follows a Continuous Time Autoregressive Moving Average (CARMA) process and specifically is able to reproduce more realistic dependence structures.  We also study the conditions of stationarity and positivity for the intensity and the strong mixing property for the increments. Furthermore we compute the likelihood, present a simulation method and discuss an estimation approach based on the autocorrelation function. 
\end{abstract}


\end{frontmatter}


\section{Introduction}
Point processes are useful mathematical models that describe the dynamics of observed event times and have been studied and applied in several fields of study from queueing theory to forestry statistics. Among the family of point processes the \cite{hawkes1971point, hawkes1971spectra} process is widely the most established and widespread model in different areas, especially in quantitative finance, actuarial science and seismology (see \citealt{ogata1988statistical} and references therein for further details). Indeed the Hawkes process is particularly interesting since it is a self-exciting process, which means that each arrival excites the intensity such that the  probability of the next arrival is increased for some period after the jump, and consequently it is well-suited to investigate, for instance, natural clustering effects and bank default in time. To show the versatility of the Hawkes process we mention a few other possible non-financial and non-insurance applications: a) social science area e.g., \cite{mohler2011self} for the modeling of urban crime and \cite{boumezoued2016population} for the population dynamics; b) social media sector as done in \cite{rizoiu2017hawkes}; and c) the modeling of disease spreading such as COVID-19 transmission as discussed in \cite{chiang2022hawkes}.

Recently the Hawkes process has gained a relevant role in financial modeling, in particular in the field of market microstructure. As a matter of fact it is used to model market activity, especially order arrivals in the limit order book \citep[e.g.,][]{bacry2013modelling, muni2017modelling, clinet2017statistical}. For a complete overview of applications of the Hawkes process in finance, we suggest the works of \cite{bacry2015hawkes} and \cite{hawkes2018hawkes}. The Hawkes process has aroused its appeal among researchers and practitioners as well as in the insurance area. Indeed, as mentioned in \cite{lesage2022hawkes}, insurance companies are interested in point processes for the quantification of regulatory capital and in managing risks (e.g., computing ruin probabilities and measuring the effect of cyber-attacks as discussed respectively in \citealt{cheng2020diffusion} and \citealt{bessy2021multivariate}).  \cite{swishchuk2021hawkes} show that the use of a Hawkes process with exponential kernel for modeling insurance claim occurrences provides an improvement over the fit of a classical Poisson model. However, they are not able to fit different empirical autocorrelation functions as exhibited in \citet[Figures 3 and 5, p. 112]{swishchuk2021hawkes}. 

As stated in \cite{errais2010affine}, the Hawkes process with  exponential kernel is Markovian and shows a good level of tractability that makes it useful for real applications in the presence of large data sets (e.g., high-frequency market data). The specification of the kernel restricts the shape of the time dependence structure of the number of jumps observed in intervals with same length. Indeed, as observed in \cite{da2014hawkes}, the autocorrelation in a Hawkes model is a decaying function of lags which is not flexible enough to represent the dependence structure observed in many data sets (e.g., wind speed data in which the exponential autocorrelation overshoots the empirical one for small lags and vice versa for large lags as documented in \citealt{benth2019non}; and, as shown in \citealt{hitaj2019levy}, mortality rates where the empirical autocorrelation function of the shock term appears to be non-monotonic).

To overcome the aforementioned drawback, in this paper we introduce a new model named CARMA(p,q)-Hawkes process. The proposed model is a Hawkes process where the intensity follows a Continuous Time Autoregressive Moving Average (CARMA) process and it is able to provide several shapes of the autocorrelation function as it removes the monotonicity constraint detected in the standard Hawkes process. The greater flexibility relies on the CARMA(p,q) component of our model, especially in the choice of the autoregressive and moving average parameters. The CARMA process has been introduced in \cite{doob1944elementary} and it is the continuous time version of the ARMA model. The advantage of the CARMA process, other than to design different shapes of autocorrelation functions, is to handle better irregular time series with respect to the ARMA process, especially for high-frequency market data, as discussed in \cite{marquardt2007multivariate} and \cite{tomasson2015some}. As a matter of fact,  the CARMA model has found many applications in the literature. Here, we list a few of these applications: a) \cite{andresen2014carma} use a CARMA(p,q) model for short and forward interest rates, while b) \cite{hitaj2019levy} employ such a model in order to capture the dynamics of the shock term in mortality modeling; c) \cite{benth2014futures} consider a non-Gaussian CARMA process for the dynamics of spot and derivative prices in electricity markets; and d) \cite{mercuri2021finite} provide formulas for the futures term structure and options written on futures in the framework of a CARMA(p,q) model driven by a time-changed Brownian motion.  As remarked in \cite{iacus2015implementation}, CARMA models have manifold interests: they can be used to describe directly the dynamics of time series and to construct the variance process in continuous time models (see \citealt{brockwell2006continuous} and \citealt{iacus2017cogarch, iacus2018discrete} for further details). Our paper presents a different type of application as we use CARMA(p,q) models for the intensity of a point process.  

In this paper, after reviewing the basic notions of the Hawkes and the CARMA processes, we introduce the CARMA(p,q)-Hawkes process and study the conditions of stationarity and positivity for the intensity, the autocorrelation function of the process and prove the strong mixing property of increments that leads us to the asymptotic distribution of the empirical autocorrelation function.


The remainder of the paper is organized as follows. Section \ref{SH} reviews the Hawkes process with exponential kernel while Section \ref{review CARMA} presents the CARMA(p,q) model in the L\'evy setting. Section \ref{CHP} introduces the CARMA(p,q)-Hawkes process. Section \ref{autocorr_CARMA} focuses on the autocorrelation function of the jumps in the proposed model and its asymptotic distribution, while Section \ref{ASS} presents a simulation and an estimation exercise. Section \ref{concl} concludes the paper. 


\section{The Hawkes Process}
\label{SH}
Point processes are useful to describe the dynamics of observed event times, i.e., a collection of realizations $\{t_i\}_{i=1}^{\infty}$ , $t_i\geq 0$ for $i=1,2, \ldots$ with $t_0:=0$  of the non-decreasing non-negative process $\left\{T_{i}\right\}_{i \geq 1 }$ called the time arrival process.
The counting process $N_t$, representing the number of events up to time $t$,  is obtained from the  time arrival process as follows:
\begin{equation}
N_{t}:=\sum_{i \geq 1} \mathbbm{1}_{\{T_i\leq t\}}
\label{genCountingProcess}
\end{equation}  
for $t \geq 0$ with associated filtration $(\mathcal{F}_t)_{t\geq 0}$ that contains the information of the counting process $N_t$ up to time $t$. 
An important quantity when dealing with a point process $N_t$ is the conditional intensity $\lambda_t$ defined as:
\begin{equation*}
\lambda_t=\lim_{\Delta\rightarrow 0^+}\frac{\mathsf{Pr}[N_{t+\Delta}-N_t=1|\mathcal{F}_{t}]}{\Delta}.
\end{equation*}
It then follows that the counting process satisfies the following properties 
\begin{equation*}
\mathsf{Pr}\left[N_{t+\Delta}-N_t= \eta \left|\mathcal{F}_t\right.\right] =
\begin{cases}
    1-\lambda_t \Delta + o\left(\Delta\right)       &  \quad \text{if } \eta = 0 \\
    \lambda_t \Delta + o\left(\Delta\right)          & \quad \text{if } \eta = 1 \\
		o\left(\Delta\right)                             & \quad \text{if } \eta > 1
\end{cases}.
\end{equation*}
The conditional intensity $\lambda_t$ of a general self-exciting process has the  following form:
\begin{equation}
\lambda_t = \mu + \int_{0}^{t}h\left(t-s\right)\mbox{d}N_s
\label{CondIntensity}
\end{equation}
with baseline intensity parameter $\mu>0$ and  (excitation) kernel function $h\left(t\right):\left[0,+\infty\right)\rightarrow \left[0,+\infty\right)$ that represents the contribution to the intensity at time $t$ that is made by an event occurred at a previous time $T_i<t$.
Following the general results about the Hawkes process in \cite{bremaud1996stability}, the stationary condition reads:
\begin{equation}
\int_0^{+\infty}h\left(t\right)\mbox{d}t<1.
\label{condst}
\end{equation}
The most used kernel is the exponential function and specifically $h\left(t\right)=\alpha e^{-\beta t}$  with $\alpha,\beta\geq 0$.
The stationary condition in \eqref{condst} implies $\beta>\alpha$ while to prove the Markovianity of the couple $\left(\lambda_t,N_t\right)$ it is enough to rewrite the intensity for any $s<t$ as
\begin{eqnarray*}
\lambda_{t}
&=&\mu+ e^{-\beta\left(t-s\right)}\int_0^{s}\alpha e^{-\beta\left(s-l\right)}\mbox{d}N_l+ \int_{s}^{t}\alpha e^{-\beta\left(t-l\right)}\mbox{d}N_l.		
\end{eqnarray*}
Observing that $\int_0^{s}\alpha e^{-\beta\left(s-l\right)}\mbox{d}N_l = \lambda_s-\mu$, thus
\begin{equation}
\lambda_{t} = \mu+ e^{-\beta\left(t-s\right)} \left(\lambda_s-\mu\right)+ \int_{s}^{t}\alpha e^{-\beta\left(t-l\right)}\mbox{d}N_l
\label{lambtdatolambs}.
\end{equation}
From \eqref{lambtdatolambs} we have that the distribution of the intensity $\lambda_t$ given the information at time $s$ depends only upon $\lambda_s$ and on the increments of the counting process over the interval $\left[s,t\right)$, which depend on the conditional intensity itself implying that the couple $\left(\lambda_t, N_t\right)$ is itself Markovian. The
intensity $\lambda_{t}$ is the solution of the following differential equation:
\begin{equation*}
\mbox{d}\lambda_t = \beta\left(\mu-\lambda_t\right)\mbox{d}t+\alpha \mbox{d}N_t, \quad \text{ with } \lambda_0=\mu.
\end{equation*}

Exploiting the Markovianity of the process $X_t:=\left(\lambda_t,N_{t}\right)$, it is possible to get the infinitesimal generator (see \citealt{errais2010affine} and \citealt{da2014hawkes} for further details) associated to a function $f:\mathbb{R}_+\times\mathbb{N}\rightarrow\mathbb{R}$ with continuous partial derivatives with respect to the first argument $\frac{\partial f}{\partial \lambda}\left(x\right)$. Starting from the definition of the infinitesimal operator for a Markov process $X_{t}$, that is:
\begin{equation*}
\mathcal{A}f:=\lim_{\Delta\rightarrow 0^+}\frac{\mathbb{E}\left[f\left(X_{t+\Delta}\right)\left|\mathcal{F}_t\right.\right]-f\left(X_{t}\right)}{\Delta},
\end{equation*} 
\cite{errais2010affine} compute the infinitesimal generator for the Hawkes process with exponential kernel that reads
\begin{equation}
\mathcal{A}f=\beta\left(\mu-\lambda_t\right)\frac{\partial f}{\partial \lambda}\left(\lambda_t,N_t\right)+\lambda_t\left[f\left(\lambda_t+\alpha,N_t+1\right)-f\left(\lambda_t,N_t\right)\right].
\end{equation}
For every function $f$ in the domain of the infinitesimal generator it is possible to build a martingale process $M_t$ with respect to the natural filtration in the following way:
\begin{equation*}
M_t=f\left(\lambda_t,N_t\right)-f\left(\lambda_0,N_0\right)-\int_0^t\mathcal{A}f\left(\lambda_s,N_s\right)\mbox{d}s
\end{equation*}
that leads to the well-known Dynkin's formula 
\begin{equation*}
\mathbb{E}\left[f\left(\lambda_t,N_t\right)\left|\mathcal{F}_s\right.\right]=f\left(\lambda_s,N_s\right)+\mathbb{E}\left[\left.\int_{s}^t\mathcal{A} f\left(\lambda_u,N_u\right) \mbox{d}u \right|\mathcal{F}_s\right], \quad \forall t > s.
\end{equation*}
The above formula for $f\equiv N_t$ has been used in \cite{da2014hawkes} to compute the moments and the autocovariance function of jump increments observed in intervals of length $\tau$ with lag $\delta$. 

\begin{proposition} Consider four time instants $t_1=t$, $t_2= t +\tau$, $t_3= t +\tau+\delta$ and $t_4= t +2\tau+\delta$, the following equalities for the Hawkes model  are obtained (see \citealt{da2014hawkes} for further details).
\begin{enumerate}
\item The long-run expected value of the number of jumps during an interval of length $\tau$ is
\begin{equation}
\mathbb{E}(\Delta_{\tau}N_{\infty}):=\lim_{t\rightarrow +\infty}\mathbb{E}[N_{t+\tau}-N_t]=\frac{\mu}{1-\frac{\alpha}{\beta}}\tau.
\label{ExpFirstMoMStdH}
\end{equation}
\item The long-run variance of the increments reads
\begin{eqnarray}
Var(\tau)&:=&\lim_{t\rightarrow +\infty}\mathbb{E}[(N_{t+\tau}-N_t)^2]-\mathbb{E}[N_{t+\tau}-N_t]^2\\ \nonumber
&=&\frac{\mu}{1-\frac{\alpha}{\beta}}(\tau(\frac{1}{1-\frac{\alpha}{\beta}})^2+(1-(\frac{1}{1-\frac{\alpha}{\beta}})^2)\frac{1-e^{\tau(\beta-\alpha)}}{\beta-\alpha}).
\label{varhawkes}
\end{eqnarray}
\item The long-run covariance of the number of arrivals for two non-overlapping intervals of length $\tau$ with lag $\delta>0$ is
\begin{eqnarray}
Cov({\tau},\delta)&:=&\lim_{t\rightarrow +\infty}\mathbb{E}[(N_{t+\tau}-N_t)(N_{t+2\tau+\delta}-N_{t+\tau+\delta})]-\mathbb{E}[N_{t+\tau}-N_t]\mathbb{E}[N_{t+2\tau+\delta}-N_{t+\tau+\delta}]\nonumber \\ 
&=&\frac{\mu\beta\alpha(2\beta-\alpha)(e^{(\alpha-\beta)\tau}-1)^2}{2(\alpha-\beta)^4}e^{(\alpha-\beta)\delta}.
\label{covhaw}
\end{eqnarray}
\item The long-run autocorrelation function of the number of jumps over intervals of length $\tau$ separated by a time lag of $\delta$ reads
\begin{equation}
Acf(\tau,\delta) = \frac{e^{-2\beta\tau}(e^{\alpha\tau}-e^{\beta\tau})^2\alpha(\alpha-2\beta)}{2(\alpha(\alpha-2\beta)(e^{(\alpha-\beta)\tau}-1)+\beta^2\tau(\alpha-\beta))}e^{(\alpha-\beta)\delta},
\label{acfhaaw}
\end{equation} 
and is always positive for $\alpha<\beta$ (stationarity condition) and exponentially decaying with the lag $\delta$.
\end{enumerate}
\end{proposition}
From \eqref{covhaw} and \eqref{acfhaaw} it is clear that the Hawkes model with exponential kernel can reproduce only strictly decreasing autocorrelation functions for varying lag values $\delta$. An interesting extension is given in  \cite{boswijk2018testing} where self-excitation is identified through the modeling of common jumps between the log price process and its own jump intensity. 

\section{L\'evy CARMA(p,q) models}\label{review CARMA}
The formal definition of a L\'evy CARMA(p,q) model \(Y_t\) with \(p > q \geq 0\) is based on the continuous version of the state-space
representation of an autoregressive moving average ARMA(p,q) model. In particular we have that
\begin{equation}
Y_t =\mathbf{b}^{\top}X_t
\label{eqCar}
\end{equation}
where \(X_t\) satisfies the following stochastic differential equation 
\begin{equation}
\mbox{d}X_t = \mathbf{A}X_{t-}\mbox{d}t+\mathbf{e}\mbox{d}Z_t,
\label{eq:CarSDE}
\end{equation} 
and $\left\{Z_{t}\right\}_{t\geq0}$ is a L\'evy process. The \(p\times p\)  matrix \(\mathbf{A}\) has the following form 
\begin{equation}
\mathbf{A}=\left[
\begin{array}{ccccc}
0 & 1 & 0 & \ldots & 0\\
0 & 0 & 1 & \ldots & 0\\
\vdots & \vdots & \vdots & \ddots & \vdots\\
0 & 0 & 0 & \ldots & 1\\
-a_p & -a_{p-1} & -a_{p-2} & \ldots & -a_1\\
\end{array}
\right]_{p\times p},
\label{Acomp}
\end{equation} 
and the \(p\times 1\)  vectors \(\mathbf{e}\) and \(\mathbf{b}\) are defined as follows 
\begin{equation}
\mathbf{e}=\left[0,0,\ldots,1\right]^\top
\label{ebold}
\end{equation}
\begin{equation}
\mathbf{b}=\left[b_0,b_1,\ldots,b_{p-1}\right]^\top
\label{bbold}
\end{equation} with \(b_{q+1}=\ldots=b_{p-1}=0\). Given a starting value for \(X_s\),
the solution of \eqref{eq:CarSDE} is \[
X_t= e^{\mathbf{A}\left(t-s\right)}X_s+\int_{s}^{t}e^{\mathbf{A}\left(t-u\right)}\mbox{d}Z_u, \quad \forall t>s
\] where
\(e^\mathbf{A}=\underset{h=0}{\stackrel{+\infty}{\sum}}\frac{1}{h!}\mathbf{A}^h\).

As reported in \citet[Section 2, p.~251]{brockwell2011estimation}, under the assumption that ensures the stationarity of $X_t$ (i.e., all eigenvalues \(\tilde{\lambda}_1,\ldots,\tilde{\lambda}_p\) of matrix
\(\mathbf{A}\) are distinct with negative real part\footnote{The eigenvalues are sorted based on their real part in an increasing order.}), the CARMA(p,q) model can be written as a summation of a finite number of continuous autoregressive models of order 1, which are also known as CAR(1) models. Specifically, 
\begin{equation}
Y_t = \mathbf{b}^\top e^{\mathbf{A}\left(t-s\right)} X_s +\int_0^{+\infty}\underset{i=1}{\stackrel{p}\sum}\left[\alpha\left(\tilde{\lambda}_i\right)e^{\tilde{\lambda}_i\left(t-u\right)}\right]\mathbbm{1}_{s\leq u\leq t}\mbox{d}Z_u 
\label{expr:Sol}
\end{equation} 
where \(\alpha\left(z\right)=\frac{b\left(z\right)}{a^{(1)}\left(z\right)}\) and the polynomials \(a\left(z\right)\) and \(b\left(z\right)\) are defined as 
\begin{equation*}
a\left(z\right):=z^p+a_1z^{p-1}+\ldots+a_p \quad \text{and} \quad b\left(z\right):=b_0+b_1z+\ldots+b_{p-1}z^{p-1}. 
\end{equation*}
Note that $a^{(1)}\left(z\right)$ is the first derivative of the polynomial $a(z)$. 

\begin{remark}
The eigenvalues of matrix $\mathbf{A}$ denoted by $\tilde{\lambda}_1,\ldots, \tilde{\lambda}_p$ are the same as the zeros of the autoregressive polynomial $a(z)$. As observed in \cite{{tsai2005note}}, the assumption that the zeros of $a(z)$  have negative real parts is a necessary condition for the stationarity of the CARMA(p,q) process $Y_t$. 
\end{remark} 
\begin{definition}
A stationary CARMA(p,q) process $Y_t$ where $Z$ is a second-order subordinator can be equivalently defined as:
\begin{equation*}\textsc{}
Y_t=\int_{-\infty}^{+\infty}h(t-u)\mbox{d}Z_u
\end{equation*}
where the function $h(t)=\mathbf{b}^{\top}e^{\mathbf{A}}\mathbb{e}\mathbbm{1}_{[0,+\infty)}(t) \mathbf{e}$ is the kernel of the CARMA(p,q) process. As $Y_t$ is independent of ${Z_s-Z_t}$, $\forall s\geq t$, the process $Y_t$ is said to be a \textsl{casual} function of the subordinator $Z_t$, also known as casual CARMA(p,q) model.
\end{definition}

In \cite{brockwell2005levy} it is shown that the function $h(u)$ can be written as 
\begin{equation}
h(u)=\sum_{i=1}^p \frac{b\left(\tilde{\lambda}_i\right)}{a^{(1)}\left(\tilde{\lambda}_i\right)}e^{\tilde{\lambda}_i\left(u\right)}\mathbbm{1}_{\{0< u <+\infty\}}.
\label{KKKK}
\end{equation} 
The positivity conditions for the kernel and for the process itself, which is for instance required for modeling the volatility using CARMA(p,q) models, have been deeply investigated in \cite{tsai2005note} and in \cite{benth2019non} for the case of positive subordinators.

\section{CARMA(p,q)-Hawkes model}
\label{CHP}
In this section, we introduce a point process where the intensity follows a CARMA(p,q) process that is a generalization of the Hawkes process with an exponential kernel.

\subsection{CARMA(p,q)-Hawkes: stationarity and positivity conditions for the intensity}

\begin{definition}
A vector process $[X_{1,t},\ldots,X_{1,p},N_{t}]^{\top}$ of dimension $p+1$ is a CARMA(p,q)-Hawkes process if the conditional intensity $\lambda_t$  of the counting process $N_t$ is a CARMA(p,q) process driven by $N_t$ and has the following form:
\begin{equation}
\lambda_t = \mu+ \mathbf{b}^\top X_t,
\label{CondIntCH}
\end{equation}
in which the baseline parameter $\mu$ is strictly positive and the vector $\mathbf{b}$ is defined as in \eqref{bbold}. The vector $X_{t}=\left[X_{1,t},\ldots,X_{1,p}\right]^{\top}$ satisfies the linear stochastic differential equation
\begin{equation}
\mbox{d}X_t = \mathbf{A}X_{t-}\mbox{d}t+\mathbf{e}\mbox{d}N_t \text{ with } X_0=\mathbf{0}
\label{eq:CarLambda}
\end{equation}
where the companion matrix $\mathbf{A}$ and the vector $\mathbf{e}$ have respectively the form in \eqref{Acomp} and \eqref{ebold}.
\end{definition}
The dynamics of the state space process $X_t$ in \eqref{eq:CarLambda} is described through a linear stochastic differential equation and it is a Markov process. Consequently, a CARMA(p,q)-Hawkes process is in turn a Markov process.
\begin{remark}
The stochastic differential equation in \eqref{eq:CarLambda} has an analytical solution given the initial condition, that is
\begin{equation}
X_t=\int_{0}^{t} e^{\mathbf{A}\left(t-s\right)}\mathbf{e} \mbox{d}N_s.
\label{18}
\end{equation}
The non-decreasing and non-negative trajectories of the counting process $N_t$ imply the positiveness of $\lambda_t$ for non-negative kernel functions.
\end{remark}


To investigate the stationary regime of a CARMA(p,q)-Hawkes model, it is necessary to determine the conditions required for a non-negative kernel, i.e., $\mathbf{b}^{\top}e^{\mathbf{A}t}\mathbf{e}\geq0, \ \forall t\geq0$. In case of a CARMA(p,q) driven by a non-negative L\'evy process the conditions of a non-negative kernel are presented in \citet[Theorem 1, p.~592]{tsai2005note}. In a similar fashion such conditions can be applied directly to our case due to the non-negative trajectories of the counting process $N_t$. Indeed, as done in \citet[Theorem 5.2]{brockwell2006continuous} for COGARCH(p,q) models, we rephrase their results for a generic CARMA(p,q)-Hawkes process when $b_0>0$ in the next proposition.

\begin{proposition}\label{propNonnegative} $ {} $
\begin{itemize}
\item[(a)] For a CARMA(p,q)-Hawkes process such that the real part of all eigenvalues of $\mathbf{A}$ is negative, the kernel function $h\left(t\right):=\mathbf{b}^{\top}e^{\mathbf{A}t}\mathbf{e}\mathbbm{1}_{\{t\geq0\}}$ is non-negative if and only if the ratio function $\frac{b\left(z\right)}{a\left(z\right)}$ is completely monotone\footnote{A function $f\left(x\right)$ defined on $\left(0,+\infty\right)$ is said to be completely monotone if and only it has derivatives of all orders and $\left(-1\right)^n\frac{\partial^n f\left(t\right)}{\left(\partial x\right)^n}\geq0$ for $n=0,1,3,\ldots$.} on $\left(0,+\infty\right)$.
\item[(b)] A sufficient condition for the kernel function of a CAR(p)-Hawkes process to be non-negative is that all eigenvalues of $\mathbf{A}$ are real and negative.
\item[(c)] A sufficient condition for the kernel function of a CAR(p)-Hawkes process to be non-negative is that if $\left(\tilde{\lambda}_{i_1}, \tilde{\lambda}_{i_1+1}\right),\ldots,\left(\tilde{\lambda}_{i_r}, \tilde{\lambda}_{i_r+1}\right)$ is a partition of the set of all pairs of complex conjugate eigenvalues of $\mathbf{A}$ (counted with multiplicity), then there exists an injective mapping $u:\left\{1,\ldots,r\right\}\rightarrow \left\{1,\ldots,p\right\}$ such that $\tilde{\lambda}_{u\left(j\right)}$ real eigenvalue of $\mathbf{A}$ satisfies $\tilde{\lambda}_{u\left(j\right)}\geq \mathsf{Re}\left(\tilde{\lambda}_{i_j}\right)$.
\item[(d)] For a non-negative kernel $h\left(t\right)$ in a CAR(p)-Hawkes process, it is necessary to find a real eigenvalue $\tilde{\lambda}_i$ such that $\tilde{\lambda}_i\geq \mathsf{Re}\left(\tilde{\lambda}_j\right)$ where $j=1,\ldots,p$ with $j\neq i$.
\item[(e)] Suppose all eigenvalues of $\mathbf{A}$ are negative real numbers sorted as follows $\tilde{\lambda}_p\leq,\ldots,\leq\tilde{\lambda}_1$ and that all the roots of  $b\left(z\right)=0$ are negative real numbers such that $\gamma_q\leq,\ldots,\leq \gamma_1<0$.
If $\sum_{i=1}^{k}\gamma_i\leq\sum_{i=1}^k \tilde{\lambda}_i$ for $1\leq k\leq q$, then the kernel of a CARMA(p,q)-Hawkes process is non-negative. 
\item[(f)] A necessary and sufficient condition for a non-negative $h\left(t\right)$ in a CARMA(2,1)-Hawkes process is that $\tilde{\lambda}_2\leq\tilde{\lambda}_1<0$ and $b_0+\tilde{\lambda}_1b_1 \geq 0$ with $b_1\geq0$.
\end{itemize}
\end{proposition}
We remark that the non-negativity requirement for the kernel implies a strictly positive intensity process $\lambda_t$ as the baseline parameter $\mu$ is strictly positive.\newline
Without loss of generality, we assume that matrix $\mathbf{A}$ is diagonalizable which corresponds to the assumption that the eigenvalues of $\mathbf{A}$ are distinct. The eigenvectors of $\mathbf{A}$ are
\[
\left[1,\tilde{\lambda}_j,\ldots,\tilde{\lambda}_{p-1}\right]^{\top}, \ j=1,\ldots,p 
\] 
used to define a $p\times p$ matrix $\mathbf{S}$ as
\begin{equation*}
\mathbf{S} :=\left[
\begin{array}{ccc}
1 & \ldots & 1\\
\tilde{\lambda}_1 & \ldots & \tilde{\lambda}_p\\
\tilde{\lambda}_1^2 & \ldots & \tilde{\lambda}_p^2\\
\vdots & & \vdots\\
\tilde{\lambda}_1^{p-1} & \ldots & \tilde{\lambda}_p^{p-1}\\
\end{array}
\right].
\end{equation*}
It follows that $\mathbf{S}$ satisfies $\mathbf{S}^{-1}\mathbf{A}\mathbf{S}=\mathsf{diag}\left(\tilde{\lambda}_1,\ldots,\tilde{\lambda}_p\right)$, a result used to prove the next proposition on the stationarity conditions for a CARMA(p,q)-Hawkes process.
\begin{proposition} \label{PropStat}
Let us consider a non-negative kernel function and suppose $\mu>0$. Then a CARMA(p,q)-Hawkes $\left(X_{1,t},\ldots,X_{p,t}, N_{t}\right)$ is a stationary process if all eigenvalues of $\mathbf{A}$ are distinct with non-negative real part and $-\mathbf{b}^{\top}\mathbf{A}^{-1}\mathbf{e}<1$. 
\end{proposition} 
\begin{proof}
For a non-negative kernel function, the stationary condition in \eqref{condst} for a CARMA(p,q)-Hawkes process becomes
\begin{eqnarray}
\int_0^{+\infty}\mathbf{b}^\top e^{\mathbf{A} t}\mathbf{e}\mbox{d}t&=&\lim_{T\rightarrow +\infty}\int_0^{T} \mathbf{b}^\top e^{\mathbf{A} t}\mathbf{e}\mbox{d}t
= \lim_{T\rightarrow +\infty} \mathbf{b}^\top \mathbf{A}^{-1}\left(e^{\mathbf{A}T}-\mathbf{I}\right)\mathbf{e},
\label{Int1}
\end{eqnarray}
where $\mathbf{I}$ is the identity matrix with dimension $p$. As $\mathbf{A}$ is diagonalizable, 
\begin{equation*}
e^{\mathbf{A}T} =\mathbf{S} e^{\mathbf{\Lambda}T} \mathbf{S}^{-1}
\end{equation*}
where $\mathbf{\Lambda}:= \mathsf{diag}\left(\tilde{\lambda}_1,\ldots, \tilde{\lambda}_p\right)$. Thus the limit in \eqref{Int1} is
\begin{eqnarray*}
\lim_{T\rightarrow +\infty} \mathbf{b}^\top \mathbf{A}^{-1}\left(e^{\mathbf{A}T}-\mathbf{I}\right)\mathbf{e}&=&  \mathbf{b}^\top \mathbf{A}^{-1}  \left[ \mathbf{S} \left(\lim_{T\rightarrow +\infty} e^{\mathbf{\Lambda}T}\right) \mathbf{S}^{-1}-\mathbf{I}\right]\mathbf{e}.
\end{eqnarray*}
Recalling that all eigenvalues of $\mathbf{A}$ have negative real part, we notice that $e^{\mathbf{\Lambda}T}$ tends to a $p \times p$ zero matrix. The integral in \eqref{Int1} becomes
\begin{equation}
\int_0^{+\infty}\mathbf{b}^\top e^{\mathbf{A} t}\mathbf{e}\mbox{d}t=-\mathbf{b}^\top \mathbf{A}^{-1} \mathbf{e}.
\end{equation}
The stationarity condition in \eqref{condst} implies $-\mathbf{b}^\top \mathbf{A}^{-1} \mathbf{e}<1$.
\end{proof}
\begin{assumption} \label{Ass1}
We shall assume for the remainder of the paper that: i) the kernel is a non-negative function and $\mu>0$; and ii) all eigenvalues of  $\mathbf{A}$ are distinct with negative real part and $\mathbf{b}^{\top}\mathbf{A}^{-1}\mathbf{e}>-1$.
\end{assumption}
For practical applications, instead of checking ex-post signs of eigenvalues of matrix $\mathbf{A}$, it is possible to enforce ex-ante the negativity of the real part for eigenvalues using some transformations on the parameters space as done, for example, in \cite{tomasson2015some}. As a CARMA(p,q)-Hawkes process is Markovian, we are able to calculate the infinitesimal operator as described in the following proposition.
\begin{proposition}
\label{Prop4}
Let $f:\mathbb{R}^{p}\times\mathbb{N}\rightarrow\mathbb{R}$ be a function with continuous partial derivatives with respect to the first $p$ arguments. Under the same conditions in Assumption \ref{Ass1}, the infinitesimal generator of function  $f$ for a CARMA(p,q)-Hawkes process is:
\begin{eqnarray}
\mathcal{A}f_t &=& \lambda_t \left[f\left(X_{1,t},\ldots,X_{p,t}+1,N_{t}+1\right)-f\left(X_{1,t},\ldots,X_{p,t},N_{t}\right)\right]\nonumber\\
&+&\overset{p-1}{\underset{i=1}{\sum}} \frac{\partial f}{\partial X_{i,t} } X_{i+1,t}+\frac{\partial f}{\partial X_{p,t} }\mathbf{A}_{[p,]}X_{t}
\label{GLobalInfGenExp}
\end{eqnarray}
where $\mathbf{A}_{[p,]}$ is the $p$-th row of the companion matrix $\mathbf{A}$ and the intensity process $\lambda_t$ is defined as in \eqref{CondIntCH}. Alternatively, the infinitesimal generator can be written as
\begin{equation}
\mathcal{A}f_t = \lambda_t \left[f\left(X_{1,t},\ldots,X_{p,t}+1,N_{t}+1\right)-f\left(X_{1,t},\ldots,X_{p,t},N_{t}\right)\right]+\nabla_p f^{\top} \mathbf{A}X_t
\label{GLobalInfGen}
\end{equation} where $\nabla_p f := \left[\frac{\partial f}{\partial X_{1,t} },\ldots \frac{\partial f}{\partial X_{p,t} }\right]^{\top}$.
\end{proposition} 
\begin{proof}
Let us consider two cases. If $N_{T+h}-N_{T}=0$, the vector $X_{T}=\left[X_{1,t},\ldots,X_{p,t}\right]^{\top}$ becomes $X_{T+h}=X^{\text{NJ}}_{T+h}$
where $X^{\text{NJ}}_{T+h}$ means no jump (NJ)  occurred in the interval $(T, T+h]$ and can be written in the following way
\begin{eqnarray*}
X^{\text{NJ}}_{T+h} 
&=& e^{\mathbf{A}\left(T+h-t_0\right)}X_{t_0}+\int_{t_0}^{T} e^{\mathbf{A}\left(T+h-t\right)}\mathbf{e}\mbox{d}N_t
\end{eqnarray*}
as the quantity $\int_{T}^{T+h} e^{\mathbf{A}\left(T+h-t\right)}\mathbf{e}\mbox{d}N_t$ is zero due to the absence of jumps in the interval $\left(T,T+h\right]$. From 
\begin{equation*}
X^{\text{NJ}}_{T+h} = e^{\mathbf{A}h}\left[e^{\mathbf{A}\left(T-t_0\right)}X_{t_0}+\int_{t_0}^{T} e^{\mathbf{A}\left(T-t\right)}\mathbf{e}\mbox{d}N_t\right]= e^{\mathbf{A}h}X_T
\end{equation*}
we have that
\begin{equation}
\underset{h\rightarrow0}{\text{lim }} X^{\text{NJ}}_{T+h} =X_T.
\label{Lim1}
\end{equation}
If $N_{T+h}-N_{T}=1$ then $X_{T+h}:=X^{\text{1J}}_{T+h}$ is computed as
\begin{equation*}
X^{\text{1J}}_{T+h} = e^{\mathbf{A}\left(T+h-t_0\right)}X_{t_0}+\int_{t_0}^{T} e^{\mathbf{A}\left(T+h-t\right)}\mathbf{e}\mbox{d}N_t+\int_{T}^{T+h} e^{\mathbf{A}\left(T+h-t\right)}\mathbf{e}\mbox{d}N_t.
\end{equation*}
Defining the jump time $T_h$ in the time interval $\left(T,T+h\right]$ we get
\[
\int_{T}^{T+h} e^{\mathbf{A}\left(T+h-t\right)}\mathbf{e}\mbox{d}N_t = e^{\mathbf{A}\left(T+h-T_h\right)}\mathbf{e}.
\] \newline 
As $\underset{h\rightarrow0}{\text{lim }} T_h=T$, we observe that
\begin{eqnarray}
\underset{h\rightarrow0}{\text{lim }} X^{\text{1J}}_{T+h} &=& \left[e^{\mathbf{A}\left(T-t_0\right)}X_{t_0}+\int_{t_0}^{T} e^{\mathbf{A}\left(T-t\right)}\mathbf{e}\mbox{d}N_t\right]+\mathbf{e}= X_T+\mathbf{e}.
\label{Lim2}
\end{eqnarray}
Note that $X_t+\mathbf{e}=\left[X_{t,1},\ldots,X_{t,p}+1\right]^\top$ and consider the following quantity:
\begin{eqnarray*}
\mathbb{E}\left[f\left(X_{1,t+h},\ldots,X_{p,t+h},N_{t+h}\right)\left|\mathcal{F}_{t}\right.\right]&=&f\left(X^{\text{NJ}}_{1,t+h},\ldots,X^{\text{NJ}}_{p,t+h},N_{t}\right)\left(1-\lambda_t h\right)\\
&+& f\left(X^{\text{1J}}_{1,t+h},\ldots,X^{\text{1J}}_{p,t+h},N_{t}+1\right)\lambda_t h +o\left(h\right). 
\end{eqnarray*}
The infinitesimal generator is:
\begin{eqnarray*}
\mathcal{A}f_t &:=& \underset{h\rightarrow0}{\text{lim }} \frac{\mathbb{E}\left[f\left(X_{1,t+h},\ldots,X_{p,t+h},N_{t+h}\right)\left|\mathcal{F}_{t}\right.\right]-f\left(X_{1,t},\ldots,X_{p,t},N_{t}\right)}{h}\\
&=& \underset{h\rightarrow0}{\text{lim }} \lambda_t\left[f\left(X^{\text{1J}}_{1,t+h},\ldots,X^{\text{1J}}_{p,t+h},N_{t}+1\right)-f\left(X^{\text{NJ}}_{1,t+h},\ldots,X^{\text{NJ}}_{p,t+h},N_{t}\right)\right]\\
&+& \underset{h\rightarrow0}{\text{lim }} \frac{f\left(X^{\text{NJ}}_{1,t+h},\ldots,X^{\text{NJ}}_{p,t+h},N_{t}\right)-f\left(N_{t},X_{1,t},\ldots,X_{p,t}\right)}{h}.
\end{eqnarray*}
\normalsize
From \eqref{Lim1} and \eqref{Lim2} we obtain
\begin{eqnarray}
\mathcal{A}f_t &:=& \lambda_t\left[f\left(X_{1,t},\ldots,X_{p,t}+1,N_{t}+1\right)-f\left(X_{1,t},\ldots,X_{p,t},N_{t}\right)\right]\nonumber\\
&+& \underset{h\rightarrow0}{\text{lim }} \frac{f\left(X^{\text{NJ}}_{1,t+h},\ldots,X^{\text{NJ}}_{p,t+h},N_{t}\right)-f\left(X_{1,t},\ldots,X_{p,t},N_{t}\right)}{h}.
\label{Lim3}
\end{eqnarray}
To compute the limit \eqref{Lim3} we use De l'H\^opital's rule
\begin{eqnarray}
\underset{h\rightarrow0}{\text{lim }} \overset{p}{\underset{i=1}{\sum}} \frac{\partial f}{\partial X^{\text{NJ}}_{i,t+h} }\frac{\partial X^{\text{NJ}}_{i,t+h}}{\partial h}&=&\underset{h\rightarrow0}{\text{lim }} \left[\frac{\partial f}{\partial X^{\text{NJ}}_{1,t+h} },\ldots \frac{\partial f}{\partial X^{\text{NJ}}_{p,t+h} }\right] \mathbf{A}e^{\mathbf{A}h}X_t\nonumber\\
&=& \overset{p-1}{\underset{i=1}{\sum}} \frac{\partial f}{\partial X_{i,t} } X_{i+1,t}+\frac{\partial f}{\partial X_{p,t} }\mathbf{A}_{[p,]}X_{t}, 
\label{FIN2}
\end{eqnarray}
and substituting \eqref{FIN2} in \eqref{Lim3}, we finally obtain the result in \eqref{GLobalInfGen}. 
\end{proof}
The conditional expected value for $f\left(X_{1,T},\ldots,X_{p,T},N_{T}\right)$ can be computed applying the Dynkin's formula:
\begin{equation}
\mathbb{E}\left[f\left(X_{1,T},\ldots,X_{p,T},N_{T}\right)\left|\mathcal{F}_{t_0}\right.\right]=f\left(X_{1,t_0},\ldots,X_{p,t_0},N_{t_0}\right)+\mathbb{E}\left[\int_{t_0}^{T}\mathcal{A}f_t\mbox{d}t\left|\mathcal{F}_{t_0}\right.\right]
\label{DF}
\end{equation}
that has a representation of the following form
\begin{equation}
\mbox{d}\mathbb{E}\left[f\left(X_{1,t},\ldots,X_{p,t},N_{t}\right)\left|\mathcal{F}_{t_0}\right.\right]=\mathbb{E}\left[\mathcal{A}f_t\left|\mathcal{F}_{t_0}\right.\right]\mbox{d}t, 
\label{ODEDynkin}
\end{equation}
with initial condition $f\left(X_{1,t_0},\ldots,X_{p,t_0},N_{t_0}\right)$. We use the infinitesimal generator \eqref{GLobalInfGen} and the result in \eqref{DF} to obtain the following proposition for the computation of the first moment of the counting process $N_t$. In the remainder of the paper, we use $\mathbb{E}_t\left[\cdot\right]:=\mathbb{E}\left[\cdot\left|\mathcal{F}_t\right.\right]$.
\begin{proposition}
\label{prop5}
Let $\tilde{\mathbf{A}}$ be a $p\times p$ companion matrix where the last row has the following structure
\begin{equation}
\tilde{\mathbf{A}}_{\left[p,\cdot\right]}=\left[b_0-a_p,b_1-a_{p-1},\ldots,b_{p-1}-a_{1}\right].
\label{AtildeMia}
\end{equation}
Under Assumption \ref{Ass1} and supposing that all eigenvalues of $\tilde{\mathbf{A}}$ are distinct with negative real part, for any $T> t_0\geq0$, the conditional first moment of the counting process is
\begin{equation}
\mathbb{E}_{t_0}\left[N_T\right] = N_{t_0}+\mu\left(1-\mathbf{b}^{\top}\tilde{\mathbf{A}}^{-1}\mathbf{e}\right)\left(T-t_0\right)+\mathbf{b}^{\top}\tilde{\mathbf{A}}^{-1}\left[e^{\tilde{\mathbf{A}}\left(T-t_0\right)}-I\right]\left[X_{t_0}+\tilde{\mathbf{A}}^{-1}\mathbf{e}\mu\right],
\label{condFirstMoMCount}
\end{equation}
while the conditional expected value of the state process $X_T$ is
\begin{equation}
\mathbb{E}_{t_0}\left[X_T\right] =e^{\tilde{\mathbf{A}}\left(T-t_0\right)}\left[X_{t_0}+\tilde{\mathbf{A}}^{-1}\mathbf{e}\mu\right]-\tilde{\mathbf{A}}^{-1}\mathbf{e}\mu.
\label{conditionalStateProcMoM}
\end{equation}
The quantities in \eqref{condFirstMoMCount} and \eqref{conditionalStateProcMoM} satisfy respectively the following ordinary differential equations:
\begin{equation}
\mbox{d}\mathbb{E}_{t_0}\left[N_{t}\right]=	\left[\mu\left(1-\mathbf{b}^\top\tilde{\mathbf{A}}^{-1}\mathbf{e}\right)+\mathbf{b}^\top e^{\tilde{\mathbf{A}}\left(t-t_0\right)}\left[X_{t_0}+\tilde{\mathbf{A}}^{-1}\mathbf{e}\mu\right]\right]\mbox{d}t
\end{equation}
and
\begin{equation}
\mbox{d}\mathbb{E}_{t_0}\left[X_{t}\right]=\left(\tilde{\mathbf{A}}\mathbb{E}_{t_0}\left[X_{t}\right]+\mu\mathbf{e}\right)\mbox{d}t
\end{equation}
with initial conditions\footnote{For $t_0=0$, then $\mathbb{E}_{t_0}\left[X_T\right] =\left(e^{\tilde{\mathbf{A}}\left(T-t_0\right)}-\mathbf{I}\right)\tilde{\mathbf{A}}^{-1}\mathbf{e}\mu$ and \newline $\mathbb{E}_{t_0}\left[N_T\right] = \mu\left(1-\mathbf{b}^{\top}\tilde{\mathbf{A}}^{-1}\mathbf{e}\right)\left(T-t_0\right)+\mathbf{b}^{\top}\tilde{\mathbf{A}}^{-1}\left[e^{\tilde{\mathbf{A}}\left(T-t_0\right)}-I\right]\tilde{\mathbf{A}}^{-1}\mathbf{e}\mu$.} $X_{t_0}$ and $N_{t_0}$. The  long-run value for $\mathbb{E}_{t_0}\left[X_T\right]$ is obtained as follows
\begin{equation}
\mathbb{E}\left[X_{\infty}\right]:=\lim_{T\rightarrow +\infty} \mathbb{E}_{t_0}\left[X_{T}\right]=-\tilde{\mathbf{A}}\mathbf{e}\mu.
\label{AsymptEX}
\end{equation}
Moreover, the expected number of events that occurs in an interval with length $\tau$, i.e., $\left(T,T+\tau\right]$, given the information at time $t_0<T$ is
\begin{eqnarray}
\mathbb{E}_{t_0}\left[\left(N_{T+\tau}-N_{T}\right)\right]
&=&\mu\left(1-\mathbf{b}^\top \tilde{\mathbf{A}}^{-1}\mathbf{e}\right)\tau+\mathbf{b}^{\top}\tilde{\mathbf{A}}^{-1}e^{\tilde{\mathbf{A}}\left(T-t_0\right)}\left(e^{\tilde{\mathbf{A}}\tau}-\mathbf{I}\right)\left(X_{t_0}+\tilde{\mathbf{A}}^{-1}\mathbf{e}\mu\right)
\label{ExpNumbOfJump}
\end{eqnarray}
and the stationary behaviour of \eqref{ExpNumbOfJump} is
\begin{equation}
\mathbb{E}\left[\Delta_{\tau}N_{\infty}\right]:=\lim_{T\rightarrow +\infty} \mathbb{E}_{t_0}\left[N_{T+\tau}-N_T\right]=\mu\left(1-\mathbf{b}^\top \tilde{\mathbf{A}}^{-1}\mathbf{e}\right)\tau, \ \ \forall \tau>0.
\label{statExpNumbOfJump}
\end{equation}

\end{proposition}
\begin{proof}
To determine the expected number of jumps in \eqref{condFirstMoMCount} we obtain first the infinitesimal generator of the function $f\left(X_{1,t},\ldots,X_{p,t},N_{t}\right)=N_t$, that is $\mathcal{A}f_t = \lambda_t$ where the conditional intensity $\lambda_t$ is defined in \eqref{eq:CarLambda}. Applying the Dynkin's formula in \eqref{ODEDynkin} we obtain the following ODE
\begin{equation}
\mbox{d}\mathbb{E}_{t_0}\left[N_{t}\right]=	\left[\mu +\mathbf{b}^\top \mathbb{E}_{t_0}\left(X_t\right)\right]\mbox{d}t.
\label{ode:Nt}
\end{equation}
Then, we compute $\mathbb{E}_{t_0}\left[X_t\right]$ that requires a system of infinitesimal generators. In particular, for $i=1,\ldots, p-1$, we have
\begin{equation*}
\mathcal{A}X_{t,i}= X_{t,i+1}
\end{equation*}
and
\begin{equation*}
\mathcal{A}X_{t,p} = \left(\mu+\mathbf{b}^\top X_t\right) + \mathbf{A}_{\left[p,\cdot\right]}X_t = \mu+\overset{p}{\underset{i=1}{\sum}}\left(b_{i-1}-a_{p+1-i}\right)X_{t,i}.
\end{equation*}
Applying \eqref{ODEDynkin}, we get
\begin{equation}
\mbox{d}\mathbb{E}_{t_0}\left[X_{t}\right]=\left(\tilde{\mathbf{A}}\mathbb{E}_{t_0}\left[X_{t}\right]+\mu\mathbf{e}\right)\mbox{d}t
\label{equation3}
\end{equation}
where $\tilde{\mathbf{A}}$ is defined in \eqref{AtildeMia}. 
With the initial condition $X_{t_0}$, the solution of the system in \eqref{equation3} is \eqref{conditionalStateProcMoM}.
Substituting \eqref{conditionalStateProcMoM} in \eqref{ode:Nt} we obtain the following ODE for the expected number of jumps
\begin{equation*}
\mbox{d}\mathbb{E}_{t_0}\left[N_{t}\right]=	\left[\mu\left(1-\mathbf{b}^\top\tilde{\mathbf{A}}^{-1}\mathbf{e}\right)+\mathbf{b}^\top e^{\tilde{\mathbf{A}}\left(t-t_0\right)}\left[X_{t_0}+\tilde{\mathbf{A}}^{-1}\mathbf{e}\mu\right]\right]\mbox{d}t
\end{equation*}
whose solution is in \eqref{condFirstMoMCount} with initial condition $N_{t_0}$.
Using the result in \eqref{condFirstMoMCount} we observe by straightforward calculations that the expected number of jumps in an interval of length $\tau$  reads as in \eqref{ExpNumbOfJump}. Due to the negativity assumption for the real part of the eigenvalues of matrix $\tilde{\mathbf{A}}$, we obtain the asymptotic behaviour in \eqref{AsymptEX} and \eqref{statExpNumbOfJump} as $\lim_{T\rightarrow+\infty}e^{\tilde{\mathbf{A}}T}=\mathbf{0}$ where $\mathbf{0}$ is a $p \times p$ zero matrix (see \eqref{veryImportant}).
\end{proof}
\begin{remark}\label{CAR_1}
The result in \eqref{statExpNumbOfJump} becomes \eqref{ExpFirstMoMStdH} if we consider a CAR(1)-Hawkes with $b_0=\alpha$ and $a_1=\beta$.
\end{remark}
Using the same arguments in \citet[proof of Proposition 4.1, p.~815]{brockwell2006continuous} , all eigenvalues of matrix $\tilde{\mathbf{A}}$ have negative real parts if for some positive integer $r\geq 1$ the following inequality holds
\begin{equation}
\left\|\mathbf{S}^{-1}\mathbf{e}\mathbf{b}^{\top}\mathbf{S}\right\|_{r}< \mathsf{Re}\left(\tilde{\lambda}_{1}\right) 
\label{eqIneq}
\end{equation}
where, in this context, $\left\|\cdot\right\|_{r}$ denotes the natural matrix norm induced by the vector $\mathbb{L}^{r}$-norm. This result comes directly from an application of the Bauer-Fike Theorem (see \citealt{bauer1960norms} for further details) since $\tilde{\mathbf{A}}$ is obtained by perturbing matrix $\mathbf{A}$ as $\tilde{\mathbf{A}}=\mathbf{A}+\mathbf{e}\mathbf{b}^{\top}$. \newline
A sufficient condition for \eqref{eqIneq} is
\begin{equation}
\frac{\sigma_{\max}\left(\mathbf{S}\right)}{\sigma_{\min}\left(\mathbf{S}\right)}\left\|\mathbf{b}\right\|_2<\mathsf{Re}\left(\tilde{\lambda}_{1}\right)
\label{aaaaa}
\end{equation}
where $\left\|\mathbf{b}\right\|_2:=\sqrt{\sum_{i=1}^{p}b^2_{i-1}}$ is the Euclidean norm of $\mathbf{b}$, $\sigma_{\max}\left(\mathbf{S}\right)$ and $\sigma_{\min}\left(\mathbf{S}\right)$ are maximal and minimal singular values of $\mathbf{S}$. In particular, we observe that
\begin{equation}
\left\|\mathbf{S}^{-1}\mathbf{e}\mathbf{b}^{\top}\mathbf{S}\right\|_{2}\leq k_2\left(\mathbf{S}\right)\left\|\mathbf{e}\mathbf{b}^{\top}\right\|_{2}
\label{a}
\end{equation}
and that $k_2\left(\mathbf{S}\right):=\left\|\mathbf{S}\right\|_2\left\|\mathbf{S}^{-1}\right\|_2$, the condition number in 2-norm, can be written as
\begin{equation}
k_2\left(\mathbf{S}\right)=\frac{\sigma_{\max}\left(\mathbf{S}\right)}{\sigma_{\min}\left(\mathbf{S}\right)}.
\label{b}
\end{equation}
Moreover, denoting with $\left\|\mathbf{e}\mathbf{b}^{\top}\right\|_{\text{F}}$ the Frobenius norm of $\mathbf{e}\mathbf{b}^{\top}$, we obtain $\left\|\mathbf{e}\mathbf{b}^{\top}\right\|_{2}\leq \left\|\mathbf{e}\mathbf{b}^{\top}\right\|_{\text{F}}$. Applying the definition of the Frobenius norm we have 
\begin{equation}
\left\|\mathbf{e}\mathbf{b}^{\top}\right\|_{2}\leq \left\|\mathbf{b}\right\|_{\text{2}},
\label{c}
\end{equation} 
and combining \eqref{a}, \eqref{b} and \eqref{c} we get
\begin{equation}
\left\|\mathbf{S}^{-1}\mathbf{e}\mathbf{b}^{\top}\mathbf{S}\right\|_{2}\leq \frac{\sigma_{\max}\left(\mathbf{S}\right)}{\sigma_{\min}\left(\mathbf{S}\right)}\left\|\mathbf{b}\right\|_2.
\end{equation}  
Thus, the inequality in \eqref{aaaaa} implies \eqref{eqIneq}.

\subsection{Simulation and Likelihood Estimation of the CARMA(p,q)-Hawkes}\label{sim_mle_CARMA-Hawkes}

We propose a simulation method for the CARMA(p,q)-Hawkes model following the same idea presented in \citet[Section 4, p.~148]{ozaki1979maximum}.

Suppose that $T_1,\ldots, T_k$, which correspond to time arrivals, are already observed. Then it is possible to simulate the next time arrival $T_{k+1}$ by generating a random number from a standard uniform distribution, i.e., $U \sim \text{Unif}\left(0,1\right)$, and by solving this equation with respect to $u$:
\begin{equation}
\ln\left(U\right) = -\int_{T_k}^{u}\lambda_t \mbox{d}t.
\end{equation}
The conditional intensity $\lambda_t$ can be replaced by \eqref{CondIntCH}
and $X_t$ can be substituted by \eqref{18}
obtaining so
\begin{equation}
  \ln\left(U\right) = -\int_{T_k}^{u} \left[\mu+ \mathbf{b}^\top \int_{0}^{t}e^{\mathbf{A}(t-s)}\mathbf{e} \mbox{d}N_s \right] \mbox{d}t.
\label{Provamia}
\end{equation}
Developing and rearranging the right-hand side of \eqref{Provamia} we have
\begin{equation*}
\begin{split}
\ln\left(U\right) &= -\mu\left(u - T_k\right) - \mathbf{b}^\top \int_{T_k}^{u} \displaystyle\sum^{k}_{i = 1} e^{\mathbf{A}(t-T_i)}\mathbf{e} \mbox{d}t \\
&=-\mu\left(u - T_k\right)   - \mathbf{b}^\top \displaystyle\sum^{k}_{i = 1} e^{\mathbf{A}(T_k - T_i)} \int_{T_k}^{u}  e^{\mathbf{A}(t-T_k)} \mbox{d}t \mathbf{e}\\
&=-\mu\left(u - T_k\right) - \mathbf{b}^\top\left[\displaystyle\sum^{k}_{i = 1}  e^{\mathbf{A}(T_k - T_i)}\right]\mathbf{A}^{-1}\left[e^{\mathbf{A}(u - T_k)}- \mathbf{I}\right]\mathbf{e}
\end{split}
\end{equation*}
where the integral in the second equality is computed using the results in \eqref{int2}. 
Defining $S\left(k\right) := \sum^{k}_{i = 1}  e^{\mathbf{A}(T_k - T_i)}$, we finally get
\begin{equation}
\ln\left(U\right) = -\mu\left(u - T_k\right) - \mathbf{b}^\top S\left(k\right) \mathbf{A}^{-1}\left[e^{\mathbf{A}(u - T_k)}- \mathbf{I}\right]\mathbf{e}.
\end{equation}
The quantity $S\left(k\right)$ can be obtained recursively as follows
\begin{equation*}
\begin{split}
S\left(1\right) &=    \mathbf{I} \\
S\left(i\right) &=    e^{\mathbf{A}(T_i- T_{i-1})} \left[S(i-1)\right] + \mathbf{I}, \quad  i \geq 2.\\
\end{split}
\end{equation*}
Note that a similar recursive expression has been obtained in \cite{ozaki1979maximum} for a Hawkes process with exponential kernel.

%
As follows we present the likelihood of a CARMA(p,q)-Hawkes model. Consider that $\mathbf{\theta} = \left(b_0,\ldots, b_q, a_1,\ldots, a_p\right)$, then the likelihood of a CARMA(p,q)-Hawkes model is given by
\begin{equation}\label{eq: likelihood}
\mathcal{L}\left(\mathbf{\theta},\mu\right)=-\int_{0}^{T_k}\lambda_t \mbox{d}t +  \int_{0}^{T_k} \ln\left(\lambda_t\right) \mbox{d}N_t.
\end{equation}
Exploiting the fact that $\int_{0}^{T_k} \ln\left(\lambda_t\right) \mbox{d}N_t = \sum_{i = 1}^{k} \ln\left(\lambda_{T_i}\right)$, then \eqref{eq: likelihood} can be written as 
\begin{equation}
\mathcal{L}\left(\mathbf{\theta},\mu\right) = -\int^{T_k}_{0} \left[\mu+\mathbf{b}^\top X_{t} \right] \mbox{d}t + \sum_{i = 1}^{k} \ln\left(\lambda_{T_i}\right)
\end{equation}
and recalling once again that $X_t$ can be expressed by \eqref{18} and rearranging the expression we have
\begin{equation}
\mathcal{L}\left(\mathbf{\theta},\mu\right) = -\mu T_k   - \mathbf{b}^\top \int^{T_k}_{0} \int_{0}^{t}  e^{\mathbf{A}(t-s)}\mathbf{e} \mbox{d}N_s \mbox{d}t  + \sum_{i = 1}^{k} \ln\left(\lambda_{T_i}\right).
\end{equation}
Working on the inner integral, the likelihood becomes
\begin{equation}
\begin{split}
\mathcal{L}\left(\mathbf{\theta},\mu\right) 
&= -\mu\left(T_k \right)  - \mathbf{b}^\top \int^{T_k}_{0} \left[\int^{T_k}_{s} e^{\mathbf{A}(t-s)}   \mbox{d}t\right]\mbox{d}N_s \mathbf{e}  + \sum_{i = 1}^{k} \ln\left(\lambda_{T_i}\right),
\end{split}
\end{equation}
while using the results in \eqref{int2} we get
\begin{equation}
\mathcal{L}\left(\mathbf{\theta},\mu\right) = -\mu T_k  - \mathbf{b}^\top \int^{T_k}_{0} \mathbf{A}^{-1}\left[e^{\mathbf{A}(T_k-s)} -  \mathbf{I}\right]\mbox{d}N_s \mathbf{e}  + \sum_{i = 1}^{k} \ln\left(\lambda_{T_i}\right).
\label{lastL}
\end{equation}
Developing the integral in \eqref{lastL}
and recalling that $S\left(k\right) := \sum^{k}_{i = 1}  e^{\mathbf{A}(T_k - T_i)}$, we finally obtain that the likelihood of a CARMA(p,q)-Hawkes model writes
\begin{equation}
 \mathcal{L}\left(\mathbf{\theta},\mu\right) = -\mu T_k  - \mathbf{b}^\top  \mathbf{A}^{-1}\displaystyle S\left(k\right)\mathbf{e} +  k \mathbf{b}^\top  \mathbf{A}^{-1} \mathbf{e}   + \sum_{i = 1}^{k} \ln\left(\lambda_{T_i}\right).
\end{equation}

\section{Autocovariance and Autocorrelation of a CARMA(p,q)-Hawkes process}\label{autocorr_CARMA}
In this section we compute the stationary autocorrelation and autocovariance functions for the number of jumps in non-overlapping time intervals of length $\tau$. To this aim we introduce some quantities that are useful to compute the asymptotic covariance of a CARMA(p,q)-Hawkes process.\newline
The first quantity we introduce is the $\frac{p\left(p+1\right)}{2}\times\frac{p\left(p+1\right)}{2}$  matrix $\tilde{\tilde{\mathbf{A}}}$ defined as follows
\footnotesize
\begin{equation}
\tilde{\tilde{\mathbf{A}}}:=\left[\begin{array}{cccccc}
D^{1}_{[p,p]} 			& U^{1,2}_{[p,p-1]} & 0_{[p,p-2]} 							& \ldots 							& \ldots 										& \ldots \\
L^{2,1}_{[p-1,p]} 	& D^{2}_{[p-1,p-1]} & U^{2,3}_{[p-1,p-2]} 			& 0_{[p-1,p-3]} 			& \ldots										& \ldots \\
\vdots							& \ddots						& \ddots										& \ddots 							& \ddots 										& \ldots \\
L^{j,1}_{[p-j+1,p]} & \ldots						& L^{j,j-1}_{[p-j+1,p-j+2]} & D^{j}_{[p-j+1,p-j+1]} & U^{j,j+1}_{[p-j+1,p-j]} & 0_{[p-j+1,p-j-1]}\\
\vdots							& \ddots						& \ddots										& \ddots 							& \ddots 										& \ldots \\
L^{p,1}_{[1,p]}     &\ldots						& \ldots										& \ldots 							& \ldots 										&D^{p}_{[1,1]}
\end{array}
\right]
\label{myAtildetilde}
\end{equation}
\normalsize
where the square matrices $D^{j}_{[p-j+1,p-j+1]}$, $j=1,\ldots,p-1$, have the following structure
\begin{equation*}
D^{j}_{[p-j+1,p-j+1]}=\left[
\begin{array}{ccccc}
0			 & 2 			& 0 		 & \ldots & 0\\
0 		 & 0 			& 1 		 & \ldots & 0\\
\vdots & \vdots & \ddots & \ddots & \vdots\\
0			 & \ldots	& \ldots & \ldots & 1\\
b_{j-1}-a_{p-j+1} & b_{j}-a_{p-j} & \ldots & \ldots & b_{p-1}-a_{1},\\
\end{array}
\right]
\end{equation*} with $D^p_{[1,1]}=2(b_{p-1}-a_1)$. Matrices $L^{j,i}_{[p-j+1,p-i+1]}$ for $j=2,\ldots,p$ and $i = 1, \ldots, j-1$ are characterized by the entries with the form
\begin{equation*}
L^{j,i}(h,l)=\left\{
\begin{array}{cl}
b_{j-2+i}-a_{p-j+1+(i-1)} & \text{ if } h = p-j+1 \text{, } l = j-i+1 \text{ and } j\neq p  \\
2\left(b_{j-2+i}-a_{p-j+1+(i-1)}\right) & \text{ if } h = p-j+1 \text{, } l = j-i+1 \text{ and } j = p  \\
0 & \text{ otherwise}
\end{array}
\right.                                                                   
\end{equation*}
while matrices $U^{i,i+1}_{[p-i+1,p-i]}$ for $i=1,\ldots,p-1$ have form
\begin{equation*}
U^{i,i+1}_{[p-i+1,p-i]}=\left[ \begin{array}{c}
\mathbf{0}_{[1,p-i]}\\
\mathbf{I}_{[p-i,p-i]}
\end{array}
\right].
\end{equation*}
Here an example of the matrix $\tilde{\tilde{\mathbf{A}}}$ for a CARMA(3,2)-Hawkes model
\begin{equation*}
\tilde{\tilde{\mathbf{A}}}=\left[
\begin{array}{cccccc}
0 			& 2 			& 0 			& 0 			& 0 			& 0	\\
0 			& 0 			& 1 			& 1 			& 0 			& 0	\\
b_0-a_3 & b_1-a_2	& b_2-a_1	& 0 			& 1 			& 0	\\
0				& 0				&	0				&	0					& 2 					& 0 \\
0				& b_0-a_3 & 0				& b_1-a_2		& b_2-a_1 		& 1 \\
0				& 0				& 2(b_0-a_3) & 0 			& 2(b_1-a_2)	& 2(b_2-a_1) \\
\end{array}
\right].
\end{equation*} 
The second quantity introduced is the $p \times \frac{p\left(p+1\right)}{2}$ matrix $\mathbf{B}$ 
defined as: 
\begin{equation}
\mathbf{B}:=\left[
\begin{array}{cccc cccccc}
b_0 		& b_1 	& \ldots & b_{p-1} & 	0   	& \ldots & \ldots  & 0 & \ldots &0\\
0 			& b_0 	& \ldots & 0			 &  b_1 	& \ldots & b_{p-1} & 0 & \ldots & 0\\
\vdots	&\ddots & \ddots & \ddots  & \ddots & \ddots & \ddots  & \ddots & \ldots & 0\\
0 		  &\ldots &	0 		 & b_0     & 0	    & \ldots & b_1     &  0     & \ldots & b_{p-1}\\       
\end{array}
\right]
\label{MyMatrixB}
\end{equation}
where the generic $i$-th row is the result of a row concatenation of $p$ vectors with dimensions $p$, $p-1$, $\ldots$, $p-i$, $\ldots 1$, respectively. The first $i-1$ vectors have zero entries except the element in position $i$ that coincides with $b_{i-1}$, the vector with dimension $p-i$ contains the elements $b_i,\ldots, b_{p-i}$ and the remaining vectors have zero entries. 


For example, in the case of a CARMA(3,2)-Hawkes model, the structure of matrix $\mathbf{B}$ reads
 \begin{eqnarray*}
\mathbf{B}=\left[
\begin{array}{cccccc}
b_0 & b_1 & b_2 & 0 	& 0 	& 0		\\
0   & b_0 & 0		& b_1 & b_2 & 0		\\
0		& 0		&	b_0 & 0		& b_1 & b_2	\\
\end{array}\right].
\end{eqnarray*}
The third quantity is the $\frac{p\left(p+1\right)}{2}\times p$ matrix $\tilde{\mathbf{C}}$ in which the entry in the $i-$ th row and in $j-$ th column has the following structure
\begin{equation}
c_{i,j}:=\left\{
\begin{array}{lcl}
0 & \text{if} & i \neq j\left(p-\frac{j-1}{2}\right) \text{ and } i\neq \frac{p\left(p+1\right)}{2}\\
\mu  & \text{if} & i = j\left(p-\frac{j-1}{2}\right) \text{ and } i\neq \frac{p\left(p+1\right)}{2}\\
b_{j-1} & \text{if} & i = \frac{p\left(p+1\right)}{2} \text{ and } j \neq p \\
2\mu+ b_{p-1} & \text{if} & i = \frac{p\left(p+1\right)}{2} \text{ and } j = p \\
\end{array}
\right. .
\label{MyMatrixC}
\end{equation}
Let $H$ be a $p\times 1$ vector. Then we define the operator $vlt\left(\cdot\right)$ as a function that transforms the $p\times p$ matrix $HH^{\top}$ into a $\frac{p\left(p+1\right)}{2}$ vector containing the lower triangular part of the product $HH^{\top}$. Specifically:
\begin{equation}
vlt\left(HH^{\top}\right):=\left[\underset{\text{p entries}}{\underbrace{H_{1}H_{1},\ldots, H_{p}H_{1}}}, \underset{\text{p-1 entries}}{\underbrace{H_{2}H_{2},\ldots, H_{p}H_{2}}}, \ldots, \underset{\text{p-i+1 entries}}{\underbrace{H_{i}H_{i},\ldots, H_{p}H_{i}}},\ldots, H_{p}H_{p} \right]^{\top}.
\label{VectrillMy}
\end{equation} 

\subsection{Conditions for existence of stationary autocovariance function}
We rewrite the quantity $\mathbb{E}_{t_0}\left[X_T X_T^{\top}\right]\mathbf{b}$ using the $vlt\left(\cdot\right)$ operator defined in \eqref{VectrillMy}.
\begin{lemma}
\label{lemmavecttrill}
The following identity holds true
\begin{equation}
\mathbb{E}_{t_0}\left[X_T X_T^{\top}\right]\mathbf{b} =\mathbf{B} vlt\left(\mathbb{E}_{t_0}\left(X_T X_T^{\top}\right)\right)
\label{IdentEXX_Bvlt}
\end{equation}
where the matrix $\mathbf{B}$ is defined in \eqref{MyMatrixB} and the operator $vlt\left(\cdot\right)$ is defined as in \eqref{VectrillMy}. Moreover:
\begin{eqnarray}
 vlt\left(\mathbb{E}_{t_0}\left(X_T X_T^{\top}\right)\right)&=& e^{\tilde{\tilde{\mathbf{A}}}\left(T-t_0\right)}vlt\left(X_{t_0}X_{t_0}^{\top}\right)+\left[e^{\tilde{\tilde{\mathbf{A}}}\left(T-t_0\right)}-\mathbf{I}\right]\tilde{\tilde{\mathbf{A}}}^{-1}\mu\left(\tilde{\mathbf{e}}-\tilde{C}\tilde{\mathbf{A}}^{-1}\mathbf{e}\right) \nonumber\\
&+& e^{\tilde{\tilde{\mathbf{A}}}T}\left[ \int_{t_0}^{T}e^{-\tilde{\tilde{\mathbf{A}}}t}\tilde{\mathbf{C}}e^{\tilde{\mathbf{A}}t}\mbox{d}t\right] e^{-\tilde{\mathbf{A}}t_0}\left[X_{t_0}+\tilde{\mathbf{A}}^{-1}\mathbf{e}\mu\right].
\label{vltXX}
\end{eqnarray}
\end{lemma}
\begin{proof}
Using the definition of matrix $\mathbf{B}$ in \eqref{MyMatrixB}, the identity in \eqref{IdentEXX_Bvlt} is straightforward. To show the result in \eqref{vltXX}, we need first to compute the infinitesimal generator for each component of $vlt\left(X_tX_t^{\top}\right)$. From the definition in \eqref{VectrillMy} we identify $p$ blocks where the dimension of each block decreases by one unit. More precisely, the $j-$th block has $p-j+1$ elements. Considering the first block (i.e., $j=1$) we have $p$ infinitesimal generators obtained applying the result in \eqref{GLobalInfGenExp} of Proposition \ref{Prop4}. For the first element in the first block, we have $\mathcal{A}X_{t,1}^2= 2 X_{t,2}X_{t,1}$.
While for the $i-$th element in the first block with $i=2,\ldots, p-1$ we get $\mathcal{A}X_{t,i}X_{t,1}= X_{t,i}X_{t,2} +X_{t,i+1}X_{t,1}$
and finally
\begin{eqnarray*}
\mathcal{A}X_{t,p}X_{t,1}&=&\lambda_t\left[\left(X_{t,p}+1\right)X_{t,1}-X_{t,p}X_{t,1}\right]+ X_{t,p}X_{t,2} +A_{\left[p,\cdot\right]}X_{t}X_{t,1}\\
&=&\mu X_{t,1}+X_{t,p}X_{t,2}+\left(\mathbf{b}^{\top}+A_{\left[p,\cdot\right]}\right)X_t X_{t,1}. 
\end{eqnarray*}
For a generic $j-$th block, we get $p-j+1$ infinitesimal generators. In particular for $i=j$ we have $\mathcal{A}X_{t,j}^2 = 2 X_{t,j}X_{t,j+1}$.
For $i=j+1,\ldots, p-1$ we have $\mathcal{A}X_{t,i}X_{t,j} =  X_{t,i}X_{t,j+1}+X_{t,j}X_{t,i+1}$
and
\begin{eqnarray*}
\mathcal{A}X_{t,p}X_{t,j} & = &  \lambda_t\left[\left(X_{t,p}+1\right)X_{t,j}-X_{t,p}X_{t,j}\right]+ X_{t,p}X_{t,j+1} +A_{\left[p,\cdot\right]}X_{t}X_{t,j}\\
& = & \mu X_{t,j}+X_{t,p}X_{t,j+1}+\left(\mathbf{b}^{\top}+A_{\left[p,\cdot\right]}\right)X_t X_{t,j}.
\end{eqnarray*}
The last block contains only one infinitesimal generator of the form
\begin{eqnarray*}
\mathcal{A}X_{t,p}^2  & = &  \lambda_t\left[\left(X_{t,p}+1\right)^2-X_{t,p}^2\right]+2 A_{\left[p,\cdot\right]}X_t X_{t,p}\\
&=& \mu+\mathbf{b}^{\top}X_{t}+2\mu X_{t,p}+2\left(\mathbf{b}^{\top}+A_{\left[p,\cdot\right]}\right)X_t X_{t,p}.
\end{eqnarray*}
Using the Dynkin's formula in \eqref{ODEDynkin} we obtain the following system of linear ODE's:
\begin{equation}
\mbox{d}vlt\left(\mathbb{E}_{t_0}\left(X_tX_{t}^{\top}\right)\right)=\left[\mu \tilde{\mathbf{e}}+ \tilde{\mathbf{C}} \mathbb{E}_{t_0}\left(X_{t}\right)+\tilde{\tilde{\mathbf{A}}}vlt\left(\mathbb{E}_{t_0}\left(X_tX_{t}^{\top}\right)\right)\right] \mbox{d} t
\label{eq:ODEXXT}
\end{equation}
where the $\frac{p\left(p+1\right)}{2}$ vector $\tilde{\mathbf{e}}$ is composed of zero entries except the last position where the element is one; $\tilde{\tilde{\mathbf{A}}}$ and $\tilde{\mathbf{C}}$ are defined in \eqref{myAtildetilde} and \eqref{MyMatrixC} respectively.\newline
The first step is to solve the ODE defined in \eqref{eq:ODEXXT} whose solution has the following form
\begin{eqnarray}
vlt\left(\mathbb{E}_{t_0}\left(X_TX_{T}^{\top}\right)\right)&=& e^{\tilde{\tilde{\mathbf{A}}}\left(T-t_0\right)}vlt\left(X_{t_0}X_{t_0}^{\top}\right)+e^{\tilde{\tilde{\mathbf{A}}}T}\int_{t_0}^{T}e^{-\tilde{\tilde{\mathbf{A}}}t}\left[\mu \tilde{\mathbf{e}}+\tilde{\mathbf{C}}\mathbb{E}_{t_0}\left(X_{t}\right)\right]\mbox{d}t\nonumber\\
&=& e^{\tilde{\tilde{\mathbf{A}}}\left(T-t_0\right)}vlt\left(X_{t_0}X_{t_0}^{\top}\right)+\left[e^{\tilde{\tilde{\mathbf{A}}}\left(T-t_0\right)}-\mathbf{I}\right]\tilde{\tilde{\mathbf{A}}}^{-1}\mu\tilde{\mathbf{e}}\nonumber\\
&+&e^{\tilde{\tilde{\mathbf{A}}}T}\int_{t_0}^{T}e^{-\tilde{\tilde{\mathbf{A}}}t}\tilde{\mathbf{C}}\mathbb{E}_{t_0}\left(X_{t}\right)\mbox{d}t.
\label{abab}
\end{eqnarray}
We also observe that\begin{eqnarray}
e^{\tilde{\tilde{\mathbf{A}}}T}\int_{t_0}^{T}e^{-\tilde{\tilde{\mathbf{A}}}t}\tilde{\mathbf{C}}\mathbb{E}_{t_0}\left(X_{t}\right)\mbox{d}t
&=& e^{\tilde{\tilde{\mathbf{A}}}T}\int_{t_0}^{T}e^{-\tilde{\tilde{\mathbf{A}}}t}\tilde{\mathbf{C}}\left[e^{\tilde{\mathbf{A}}\left(t-t_0\right)}\left[X_{t_0}+\tilde{\mathbf{A}}^{-1}\mathbf{e}\mu\right]-\tilde{\mathbf{A}}^{-1}\mathbf{e}\mu\right]\mbox{d}t\nonumber\\
&=&  e^{\tilde{\tilde{\mathbf{A}}}T}\int_{t_0}^{T}e^{-\tilde{\tilde{\mathbf{A}}}t}\tilde{\mathbf{C}}e^{\tilde{\mathbf{A}}t}\mbox{d}t e^{-\tilde{\mathbf{A}}t_0}\left[X_{t_0}+\tilde{\mathbf{A}}^{-1}\mathbf{e}\mu\right] \nonumber\\
&-& \left[e^{\tilde{\tilde{\mathbf{A}}}\left(T-t_0\right)}-\mathbf{I}\right]\tilde{\tilde{\mathbf{A}}}^{-1}\tilde{C}\tilde{\mathbf{A}}^{-1}\mathbf{e}\mu.
\label{ababa}
\end{eqnarray}
Substituting \eqref{ababa} into \eqref{abab} we obtain the result in \eqref{vltXX}.
\end{proof}
\ref{proofs} contains the proofs of the following two propositions on the variance and covariance of the number of jumps that occur in two non-overlapping time intervals of the same length for a CARMA(p,q)-Hawkes model.
\begin{proposition}
\label{PropAutocovCARMAHAWKES}
Under Assumption \ref{Ass1} and supposing that all eigenvalues of $\tilde{\mathbf{A}}$ and $\tilde{\tilde{\mathbf{A}}}$ have negative real parts, the long-run covariance $Cov\left(\tau,\delta\right)$ defined as in \eqref{covhaw} for a CARMA(p,q)-Hawkes process has the following form:
\begin{equation}
Cov\left(\tau,\delta\right)=\mathbf{b}^{\top}\tilde{\mathbf{A}}^{-1}\left[e^{\tilde{\mathbf{A}}\tau}-\mathbf{I}\right]e^{\tilde{\mathbf{A}}\delta} g_{\infty}\left(\tau\right)
\label{AutoCovCARMAHAWKES}
\end{equation}
where $g_{\infty}\left(\tau\right)$ is defined as
\begin{equation}
g_{\infty}\left(\tau\right) :=\left(\mathbf{I}-e^{\tilde{\mathbf{A}}\tau}\right) \tilde{\mathbf{A}}^{-1}\mu \left[\mathbf{e}\mathbf{b}^\top \tilde{\mathbf{A}}^{-1}\mathbf{e}-\mathbf{e} +\tilde{\mathbf{A}}^{-1}\mathbf{e}\mu\left(\mathbf{b}^{\top}\tilde{\mathbf{A}}^{-1}\mathbf{e}\right) +\mathbf{B}\tilde{\tilde{\mathbf{A}}}^{-1}\left(\tilde{\mathbf{e}}-\tilde{\mathbf{C}}\tilde{\mathbf{A}}^{-1}\mathbf{e}\right)\right].
\label{ginftytau}
\end{equation}
\end{proposition}
\begin{proposition} 
\label{PropVarCARMAHAWKES}
Under the same assumptions as in Proposition \ref{PropAutocovCARMAHAWKES}, the long-run variance $Var\left(\tau\right)$ of the number of jumps in a interval with length $\tau$ for a CARMA(p,q)-Hawkes process, defined as in \eqref{varhawkes}, has the following form:
\begin{eqnarray}
Var\left(\tau\right)&=&\left(1-\mathbf{b}^{\top}\tilde{\mathbf{A}}^{-1}\mathbf{e}\right)\left(1-2\mathbf{b}^{\top}\tilde{\mathbf{A}}^{-1}\mathbf{e}\right)\mu\tau+2\mathbf{b}^{\top}\tilde{\mathbf{A}}^{-1}\tilde{\mathbf{A}}^{-1}\mathbf{e}\tau\mu^2\left(\mathbf{b}^{\top}\tilde{\mathbf{A}}^{-1}\mathbf{e}\right)\nonumber\\
&+&2\mathbf{b}^{\top}\tilde{\mathbf{A}}^{-1}\mathbf{B}\tilde{\tilde{\mathbf{A}}}^{-1}\mu\left(\tilde{\mathbf{e}}-\tilde{\mathbf{C}}\tilde{\mathbf{A}}^{-1}\mathbf{e}\right)\tau-2\mathbf{b}^{\top}\tilde{\mathbf{A}}^{-1}\left[e^{\tilde{\mathbf{A}}\tau}-\mathbf{I}\right]h_{\infty}\left(0\right)
\label{AsymptoticVarianceJumps}
\end{eqnarray}
where $h_{\infty}\left(0\right)$ is defined as
\begin{equation}
h_{\infty}\left(0\right):=-\tilde{\mathbf{A}}^{-1}\mathbf{e}\mu\left(1-\mathbf{b}^{\top}\tilde{\mathbf{A}}^{-1}\mathbf{e}\right)+\tilde{\mathbf{A}}^{-1}\tilde{\mathbf{A}}^{-1}\mathbf{e}\mu^2\mathbf{b}^{\top}\tilde{\mathbf{A}}^{-1}\mathbf{e}+\tilde{\mathbf{A}}^{-1}\mathbf{B}\tilde{\tilde{\mathbf{A}}}^{-1}\mu\left(\tilde{\mathbf{e}}-\tilde{\mathbf{C}}\tilde{\mathbf{A}}^{-1}\mathbf{e}\right).
\end{equation}
\end{proposition}

\begin{remark}
Combining the results in Propositions \ref{PropAutocovCARMAHAWKES} and \ref{PropVarCARMAHAWKES}, we determine the asymptotic autocorrelation function of number of jumps in non-overlapping time intervals of length $\tau$, i.e., $\rho_{\tau}\left(d\right)$, for a CARMA(p,q)-Hawkes in a closed-form formula:
\begin{equation}
\rho_{\tau}\left(d\right) =\frac{Cov\left(\tau,d-1\right)}{Var\left(\tau\right)}, \ \ d=1,2,\ldots
\label{acfCARMAHAWKES}
\end{equation}
where $d$ denotes the lag order.
\end{remark}

\subsection{Strong mixing property for the increments of a CARMA(p,q)-Hawkes and asymptotic distribution of the autocorrelation function}\label{smpCARMA-Hawkes}

The asymptotic distribution of the autocorrelation function of a CARMA(p,q)-Hawkes process can be easily obtained if we show that the increments of the process are strongly mixing.
\begin{definition}
Let $\left(\Omega,\mathcal{F},\mathbb{P}\right)$ be a probability space and $\mathcal{A},\mathcal{B}$ two sub $\sigma-$algebras of $\mathcal{F}$. The strong-mixing coefficient is defined as:
\begin{equation}
\alpha\left(\mathcal{A},\mathcal{B}\right):= \sup\left\{\left|\mathbb{P}\left(A\cap B\right)-\mathbb{P}\left(A\right)\mathbb{P}\left(B\right)\right|A\in \mathcal{A}, B \in \mathcal{B}\right\}.
\label{strongCoeff}
\end{equation} 
\end{definition}
Following \cite{poinas2019mixing}, the quantity in \eqref{strongCoeff} can be reformulated for a point process $N_t$ in the following way:
\begin{equation}
\alpha_{N}\left(r\right):=\sup_{t\in \mathbb{R}} \ \ \alpha\left(\xi_{-\infty}^t,\xi_{t+r}^{\infty}\right)
\end{equation}
where $\xi_a^b$ denotes the $\sigma-$algebra generated by the cylinder sets on the interval $\left(a,b\right]$\footnote{
Let $N$ be a counting process defined as a map from a probability space $\left(\Omega,\mathcal{F},\mathbb{P}\right)$ to a measurable space $\left(\mathbb{M},\mathcal{M}\right)$ of locally finite counting measures on $\Omega$. Then the $\sigma-$algebra $\xi_a^b$ is defined as:
\[
\xi_a^b:=\sigma\left(\left\{N\in\mathbb{M}:N\left(A\right)=n\right\}; A \in \mathcal{B}\left(\left(a,b\right]\right), n \in \mathbb{N}\right).
\]
}. Considering the  sequence $\left(\Delta_1 N_k\right)_{k\in \mathbb{Z}}$ where $\Delta_1 N_k:=N_{k+1}-N_{k}$ is the number of jumps in the interval of length 1 and extremes $k$,  $k+1$, the strong-mixing coefficient has the form
\begin{equation}
\alpha_{\Delta_1 N}\left(r\right) := \sup_{n\in \mathbb{Z}} \ \ \alpha\left(\mathcal{F}_{-\infty}^{n}, \mathcal{F}_{n+r}^{\infty}\right)
\end{equation}
where $\mathcal{F}_{a}^{b}$ is the $\sigma-$algebra generated by the sequence $\left(\Delta_1 N_k\right)_{a\leq k\leq b}$. If $\alpha_{N}\left(r\right)\rightarrow 0$ (respectively $\alpha_{\Delta_1 N_k}\left(r\right)\rightarrow 0$) as $r\rightarrow +\infty$, the point process $N_t$ (respectively $\Delta_1 N_k$) is said to be strongly-mixing.\\
Using  Theorem 1 in \cite{cheysson2020strong}, we obtain the following proposition. 
\begin{proposition}
\label{PropStrongCH}
A CARMA(p,q)-Hawkes  process satisfying Assumption \ref{Ass1} is strongly mixing with exponential rate.
\end{proposition}
\begin{proof}
We first prove the existence of a positive constant $a_0>0$ such that the kernel function satisfies the condition
\begin{equation}
\int_{\mathbb{R}}e^{a_0\left|t\right|}h\left(t\right)\mbox{d}t
<+\infty.
\label{converg}
\end{equation}
We notice that Assumption \ref{Ass1} implies that
\begin{eqnarray*}
\int_{\mathbb{R}}e^{a_0\left|t\right|}h\left(t\right)\mbox{d}t&=&\mathbf{b}^{\top}\int_0^{+\infty}e^{a_0t}e^{\mathbf{A}t}\mbox{d}t\mathbf{e}=\mathbf{b}^{\top} \mathbf{S}\int_0^{+\infty}e^{a_0t}e^{\mathbf{\Lambda}t}\mbox{d}t \mathbf{S}^{-1}\mathbf{e}. 
\end{eqnarray*}
Choosing $a_0 \in \left(0,\left|\mathsf{Re}\left(\lambda_1\right)\right|\right)$ the condition in \eqref{converg} is ensured and thus we can apply the result in Theorem 1 proved by \cite{cheysson2020strong}, and the strong-mixing coefficient results to be $\alpha_N\left(r\right)=O\left(e^{-ar}\right)$ where $a\in\left(0,a_0\right)$. 
\end{proof}
As shown in \cite{cheysson2020strong}, we have that $\alpha_{\Delta_1 N}\left(r\right)\leq \alpha_{N}\left(r\right)$ and the result in Proposition \ref{PropStrongCH} implies that the sequence $\left(\Delta_1 N_k\right)_{k \in \mathbb{Z}}$ is strongly mixing. This result is useful to determine the asymptotic distribution of the sample autocovariance and autocorrelation functions associated to the sequence $\left(\Delta_1 N_k\right)_{k \in \mathbb{Z}}$.
Following the result in \cite{ibragimov1975independent}, we obtain the  following result for the  asymptotic distribution of the  sample mean, the sample variance and the sample autocovariance function. 
\begin{proposition}
\label{AsymptoticDistribution}
Let $\left(N_t\right)_{t\geq0}$ be a stationary CARMA(p,q)-Hawkes process that satisfies the assumptions in Proposition \ref{PropStrongCH}. We assume the existence of a positive constant $\phi$ such that $\mathbb{E}\left[\left(\Delta_1 N_1\right)^{4+\phi}\right]<+\infty$. Denoting with
\begin{equation*}
V_k:=\left[\begin{array}{c}
\Delta_1 N_k \\
\left(\Delta_1 N_k-\mathbb{E}\left(\Delta_1 N_{\infty}\right)\right)^2\\
\left(\Delta_1 N_k-\mathbb{E}\left(\Delta_1 N_{\infty}\right)\right)\left(\Delta_1 N_{k+1}-\mathbb{E}\left(\Delta_1 N_{\infty}\right)\right)\\
\vdots\\
\left(\Delta_1 N_k-\mathbb{E}\left(\Delta_1 N_{\infty}\right)\right)\left(\Delta_1 N_{k+d}-\mathbb{E}\left(\Delta_1 N_{\infty}\right)\right)
\end{array}
\right], \text{ with }k=1,...,n  \text{ and } d<n
\end{equation*}
as $n\rightarrow +\infty$, we have:
\begin{equation}
\sqrt{n}\left(\frac{1}{n}\sum_{k=1}^{n}V_k-\left(\begin{array}{c}
\mathbb{E}\left(\Delta_1 N_{\infty}\right)\\
Var\left(\Delta_1 N_{\infty}\right)\\
\text{Acv}\left(1\right)\\
\vdots\\
Acv\left(d\right)
\end{array}\right)\right) \rightarrow \mathcal{N}_{d+2}\left(\mathbf{0},\Sigma\right)
\label{AcvasymptDistr}
\end{equation}
 where $Acv\left(d\right):=Cov(1,d-1)$ and
\begin{equation}
\Sigma := \mathbb{E}\left(V_1 V_1^{\top}\right)+2\sum_{k=2}^{+\infty} \mathbb{E}\left(V_1 V_k^{\top}\right).
\end{equation}
\end{proposition}
\begin{proof}
The proof is quite standard and is an application of Theorem 18.5.3 in \cite{ibragimov1975independent} and Cram\'er-Wold device. Denoting with 
\begin{equation*}
\vartheta := \left[\mathbb{E}\left(\Delta_1 N_{\infty}\right),
Var\left(\Delta_1 N_{\infty}\right),
Acv\left(1\right),
\ldots,
Acv\left(d\right)
\right]^{\top}
\end{equation*}  
we apply Theorem 18.5.3 in \cite{ibragimov1975independent} to the linear combination $\left(c^{\top}V_{k}\right)_{k=1,2,\ldots n}$ where $c$ is a generic $d+2$ real vector such that $c\top\Sigma c>0$. Since the strong mixing property is preserved under linear transformations as well as the rate we have
\begin{equation*}
\sqrt{n}\left(\frac{1}{n}\sum_{k=1}^{n}c^{\top}V_{k}- c^{\top} \vartheta\right)\rightarrow \mathcal{N}\left(0,c^{\top}\mathbb{E}\left(V_{1}
V_1^{\top}\right)c+2\sum_{k=1}^{+\infty}c^{\top}\mathbb{E}\left(V_{1}
V_k^{\top}\right)c\right),
\end{equation*}
that is
\begin{equation*}
\sqrt{n}\left(\frac{1}{n}\sum_{k=1}^{n}c^{\top}V_{k}- c^{\top} \vartheta\right)\rightarrow \mathcal{N}\left(0,c\top\Sigma c\right).
\end{equation*}
Applying Cram\'er-Wold device we obtain the asymptotic behavior in \eqref{AcvasymptDistr}.     
\end{proof}
Applying the Delta method, it is possible to use the result in Proposition \ref{AsymptoticDistribution} to obtain the asymptotic distribution of the autocorrelation function.

\section{Application to simulated series}
\label{ASS}

This section examines the time series of the counting process $N_t$ generated by the simulation of two different stochastic processes: the standard Hawkes process and the CARMA(3,1)-Hawkes process. We show the behaviour of the autocorrelation function and its $95\%$ confidence interval obtained applying the results of Subsection~\ref{smpCARMA-Hawkes}. Furthermore, we investigate the estimation of parameters by means of the Maximum Likelihood Estimation (MLE) method described in Subsection \ref{sim_mle_CARMA-Hawkes} and the Moment Matching Estimation (MME) method which we describe below. 

Consider a sequence of empirical observations for the increments of a counting process $\left(\Delta_{\tau}N_k\right)_{k = 1, \dots, n}$, then the MME method is composed of two steps. The first step is to compute the least squares estimator:

\begin{equation*}
\bm{\hat{\theta}}_{n, \tau} := \operatornamewithlimits{arg min}_{\bm{\hat{\theta}}_{n, \tau} \in \Theta \subseteq \mathbb{R}^{p+q+1}} M\left(\hat{\rho}_{n, \tau}, \theta \right)
\end{equation*}
where $\Theta$ is a subset of $\mathbb{R}^{p+q+1}$ such that the stationary condition is guaranteed, the kernel function is non-negative defined and higher order moments of a CARMA(p,q)-Hawkes process exist. For a fixed $m \geq p + q + 1$, $M: \mathbb{R}_{+}^{m}\times \Theta \to \mathbb{R}$ is defined as: 
\begin{equation*}
M\left(\hat{\rho}_{n, \tau}, \theta \right):= \displaystyle\sum_{d = 1}^{m} \left(\hat{\rho}_{n, \tau} (d) - \rho_\tau(d)\right)^2
\end{equation*}
in which $d$ denotes the lag order, $\hat{\rho}_{n, \tau} (d)$ represents a vector containing the empirical autocorrelations and $\rho_\tau(d)$ is a vector of theoretical autocorrelations described in Section~\ref{autocorr_CARMA}. The vector $\theta$ includes only the autoregressive ($a_1, \dots, a_p$) and moving average ($b_0, \dots, b_q$) parameters. Once obtained the autoregressive and moving average parameters, the parameter $\mu$  can be estimated, which corresponds to the second step, from Equation \eqref{statExpNumbOfJump} using the empirical first moment of $\Delta_{\tau} N_t$ with $\tau = 1$.

Note that all chosen parameters for the Hawkes process and the CARMA(3,1)-Hawkes process ensure two conditions: the stationarity of the process (see Section~\ref{SH}) and the existence of the asymptotic autocorrelation function (see Section~\ref{autocorr_CARMA}). 

\subsection{The Hawkes process}\label{sim_Hawkes}

We simulate the Hawkes process with the following parameters: $\mu=0.2$, $a_1=0.7$ and $b_0=0.5$. For the sake of clarity, we recall that the Hawkes process can be seen also as a CAR(1)-Hawkes process as mentioned in the Remark~\ref{CAR_1}.

\begin{figure}[!h]
\centering
\begin{subfigure}[b]{0.45\textwidth}
         \centering
         \includegraphics[width=\textwidth]{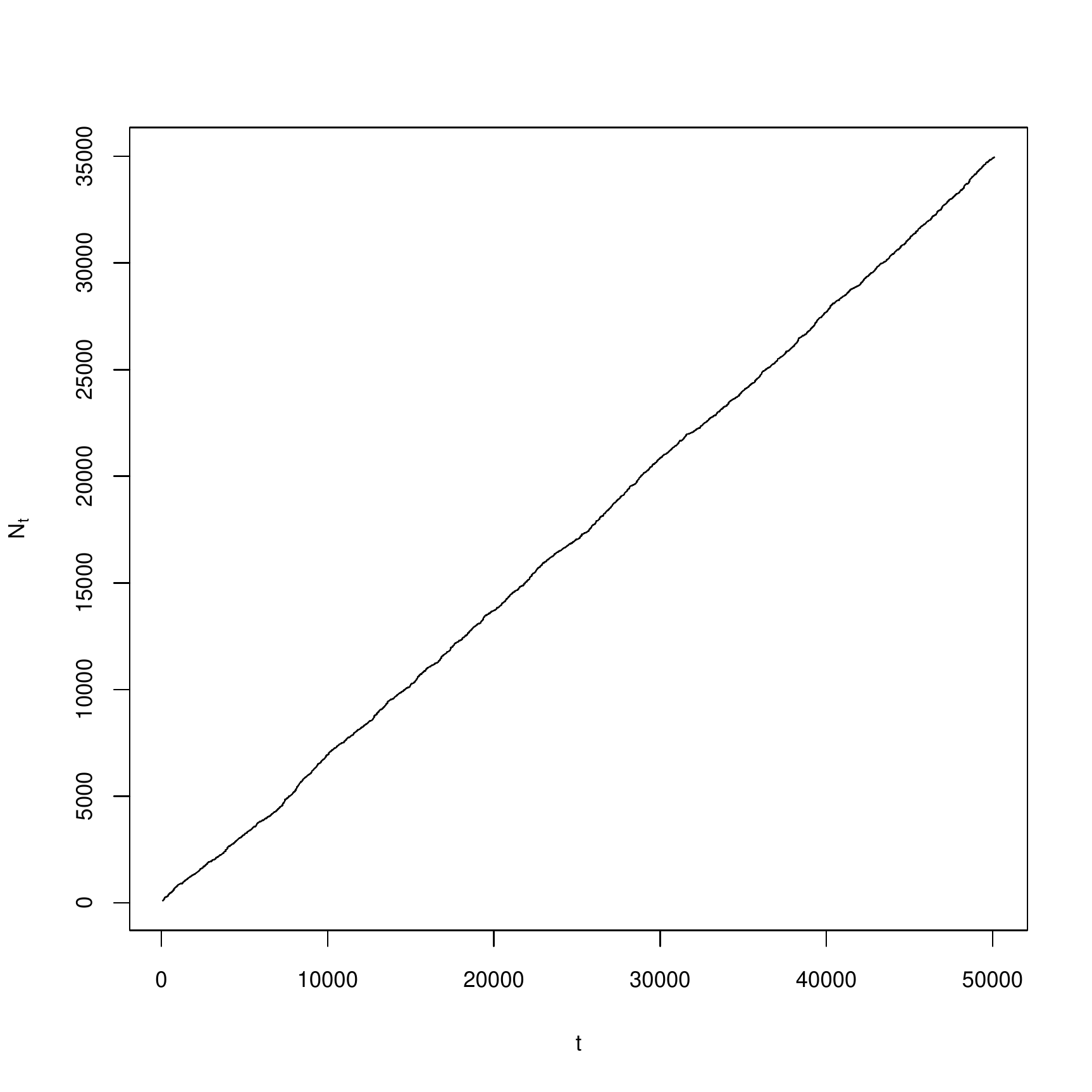}
				 \caption{Trajectory of the counting process $N_t$}
         \label{TrajectHawkes}
\end{subfigure}
\hfill
\begin{subfigure}[b]{0.45\textwidth}
         \centering
         \includegraphics[width=\textwidth]{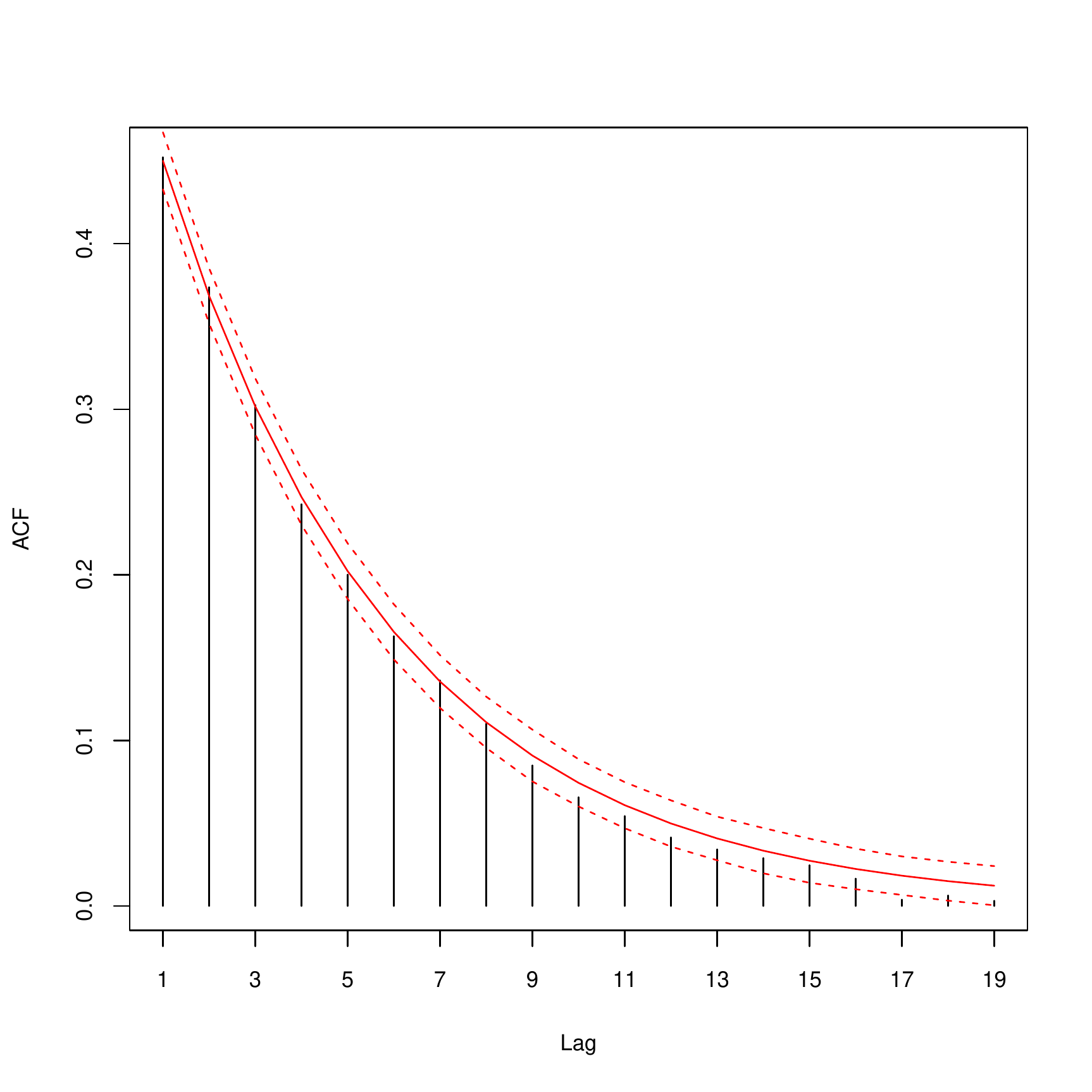}
         \caption{Autocorrelation function}
         \label{AcfHawkes}
\end{subfigure}
\caption{Trajectory of the counting process $N_t$ and autocorrelation function (ACF) using a Hawkes process. Input parameters: $\mu=0.2$, $a_1=0.7$ and $b_0=0.5$.}
\label{fig:Hawkes}
\end{figure}

Figure \ref{TrajectHawkes} shows a simulated trajectory of the counting process $N_t$, while Figure \ref{AcfHawkes} exhibits the autocorrelation function at different lags. From Figure \ref{AcfHawkes} we  observe that all empirical autocorrelations belong to the 95\% confidence interval and, as expected, the ACF depicts an exponential monotonic decay (decreasing) behaviour which is typical of the Hawkes processes. 

Table \ref{tab: MLE_Hawkes} exhibits the MLE estimates for the parameters and the number of occurred events  for different levels of final time $T$, whereas Table \ref{tab: MME_Hawkes} shows the estimated parameters using the MME method.
%

\begin{table}[!h]
\centering
\begin{tabular}{@{}ccccc@{}}
\toprule
$\hat{\mu}$ & $\hat{a}_1$  & $\hat{b}_0$  & $N_t$ &  $T$  \\ \hline
0.2101 & 0.7270 & 0.4883 &  3199   &  5000\\
0.2014 & 0.7389 & 0.5211 & 10240   & 15000\\
0.1987 & 0.7054 & 0.4997 & 17042   & 25000\\
0.2011 & 0.7028 & 0.5004 & 34914   & 50000\\
\bottomrule
\end{tabular}
\caption{Parameter estimates from MLE and number of occurred events $N_T$ for a Hawkes process. True parameters are $\mu=0.2$, $a_1=0.7$ and $b_0=0.5$.}\label{tab: MLE_Hawkes}
\end{table}

\begin{table}[!h]
\centering
\begin{tabular}{@{}cccc@{}}
\toprule
$\hat{\mu}$ & $\hat{a}_1$   & $\hat{b}_0$ & $T$  \\ \hline
0.2440 & 0.7540 & 0.4877 & 5000\\
0.2121 & 0.7127 & 0.4954  & 15000\\
0.2044 & 0.7016 & 0.4951 & 25000\\
0.1992 & 0.7042 & 0.4990  & 50000\\
\bottomrule
\end{tabular}
\caption{Parameter estimates from MME for a Hawkes process. True parameters are $\mu=0.2$, $a_1=0.7$ and $b_0=0.5$.}\label{tab: MME_Hawkes}
\end{table}

\subsection{CARMA(3,1)-Hawkes process}

We simulate the counting process $N_t$ using a CARMA(3,1)-Hawkes model with the following parameters: $\mu =0.30$, $a_1 = 1.3$, $a_2 = 0.34+\pi^2/4 \approx 2.807 $, $a_3 = 0.025+0.025\pi^2 \approx 0.2717$, $b_0 = 0.2$ and $b_1 = 0.3$. This set of parameters ensures the stationary conditions since $-\mathbf{b}^{\top}\mathbf{A}^{-1}\mathbf{e}\approx 0.7359973 < 1$ and the largest eigenvalue of $\mathbf{A}$ is $\tilde{\lambda}_1\approx-0.1012$. We can easily verify the condition for the existence of $\mathbb{E}\left(\Delta_{\tau}N_{\infty}\right)$ in \eqref{statExpNumbOfJump}. Indeed, the real part of the largest eigenvalue of $\tilde{\mathbf{A}}$ is negative $(-0.02903)$. Moreover, it is possible to verify that all eigenvalues of $\tilde{\tilde{\mathbf{A}}}$ have negative real part, thus the long-run autocorrelation function exists (the real part of the largest eigenvalue of $\tilde{\tilde{\mathbf{A}}}$ is $-0.0290$).
In order to analyze the nonnegativity of the kernel $h\left(t\right)$, Figure \ref{Kernel31} presents the behaviour of $h\left(t\right)$ with $t \in \left[0,30\right]$.  From \eqref{KKKK} the tail behaviour of $h\left(t\right)$ is proportional to $e^{-0.1012t}$ (i.e., $h\left(t\right)\sim 1.9800e^{-0.1012t} $ as $t\rightarrow +\infty$).
 
\begin{figure}[!h]
\centering
\includegraphics[scale=0.45]{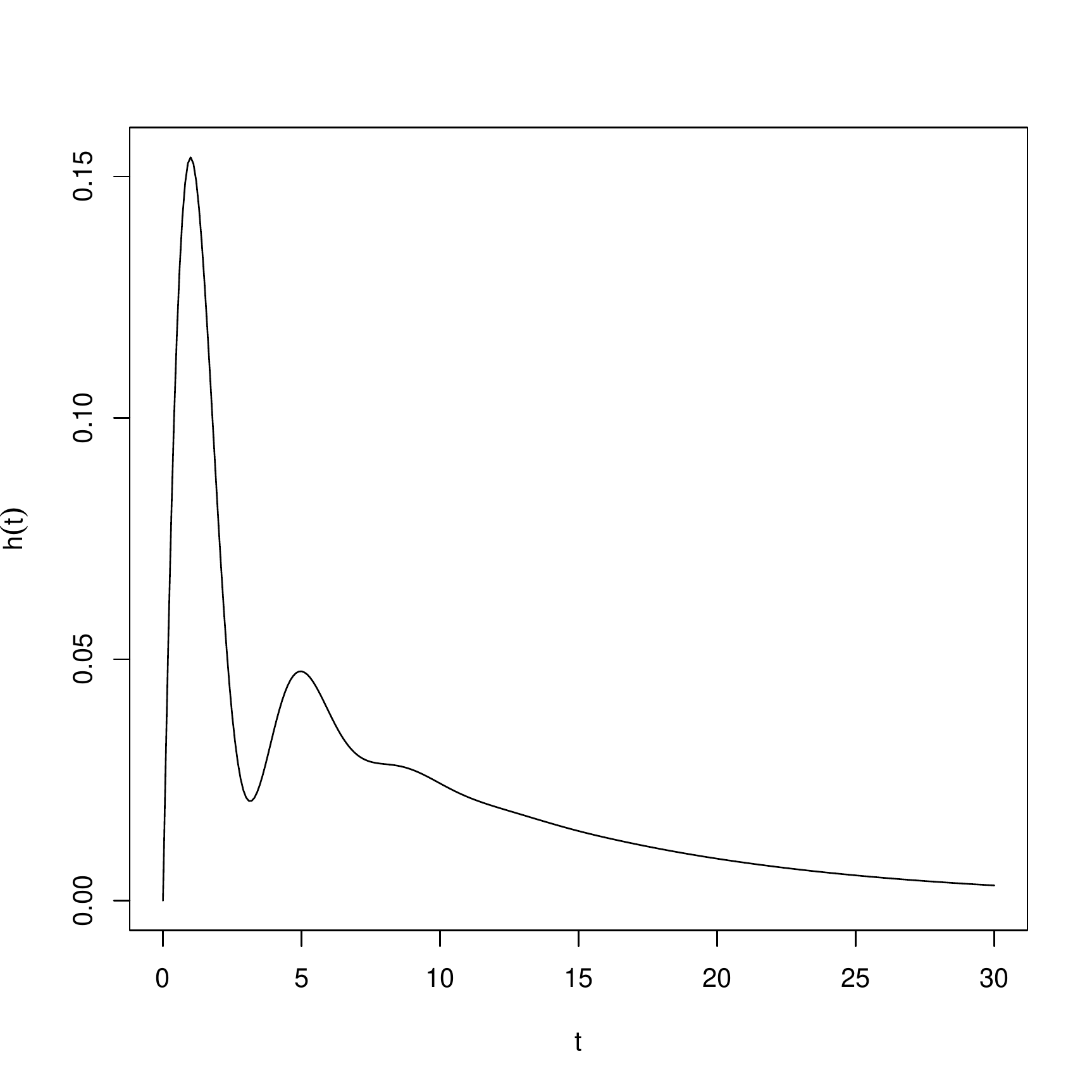}
\caption{Kernel function of a CARMA(3,1)-Hawkes. Input parameters: $\mu =0.30$, $a_1 = 1.3$, $a_2 = 0.34+\pi^2/4 \approx 2.807 $, $a_3 = 0.025+0.025\pi^2 \approx 0.2717$, $b_0 = 0.2$ and $b_1 = 0.3$.}
\label{Kernel31}
\end{figure}

Figure \ref{TrajectCARMA31} exhibits a simulated trajectory of the counting process $N_t$, while Figure \ref{AcfCARMA31} displays the ACF at different lags. Figure \ref{AcfCARMA31} is a clear example of the fact that a CARMA(3,1)-Hawkes process can accommodate different shapes of autocorrelation structures rather than an exponential monotonic decay behaviour as it is the case of a standard Hawkes process. 

\begin{figure}[!h]
\centering
\begin{subfigure}[b]{0.45\textwidth}
   \centering
   \includegraphics[width=\textwidth]{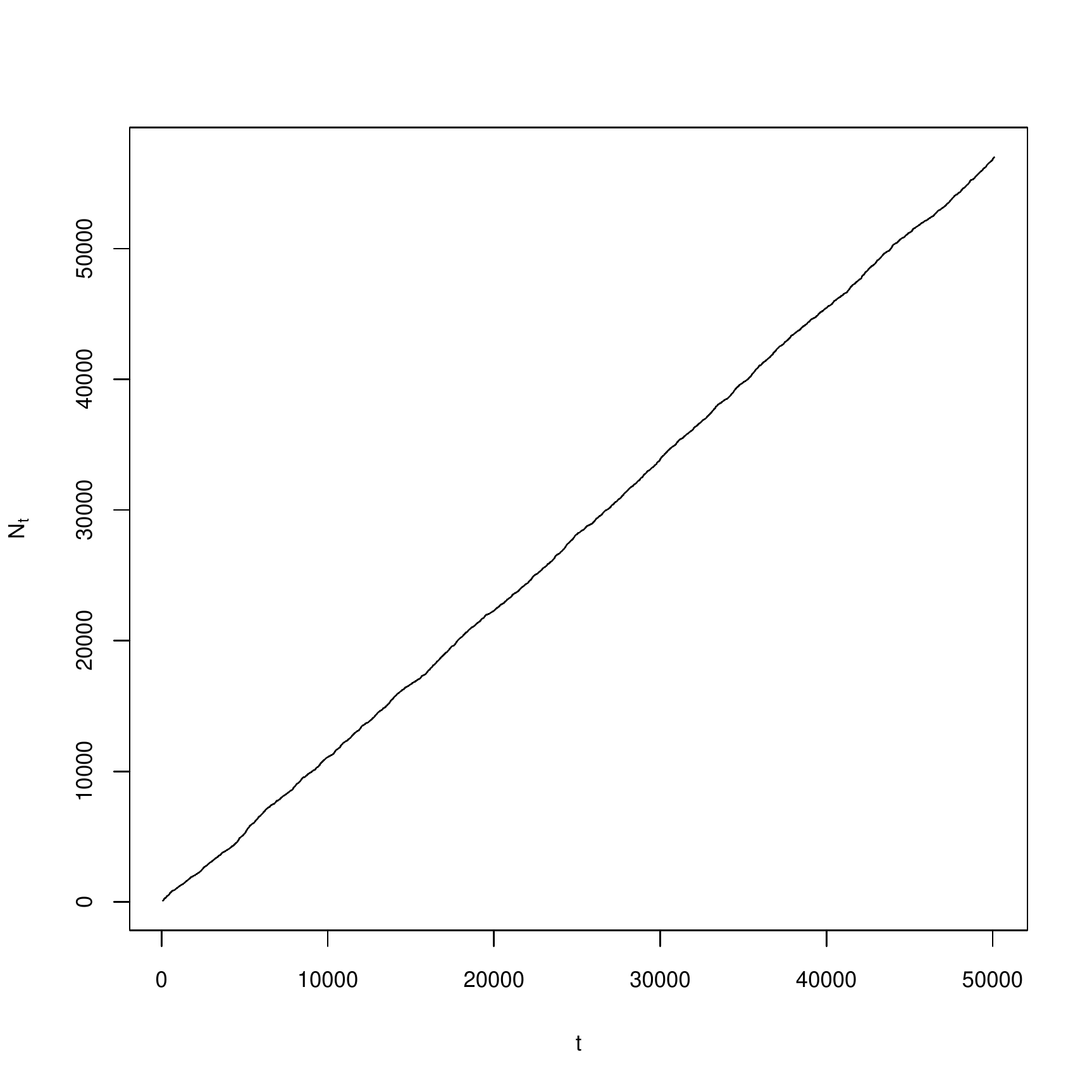}
   \caption{Trajectory of the counting process $N_t$}
	 \label{TrajectCARMA31}
\end{subfigure}
\hfill
\begin{subfigure}[b]{0.45\textwidth}
   \centering
   \includegraphics[width=\textwidth]{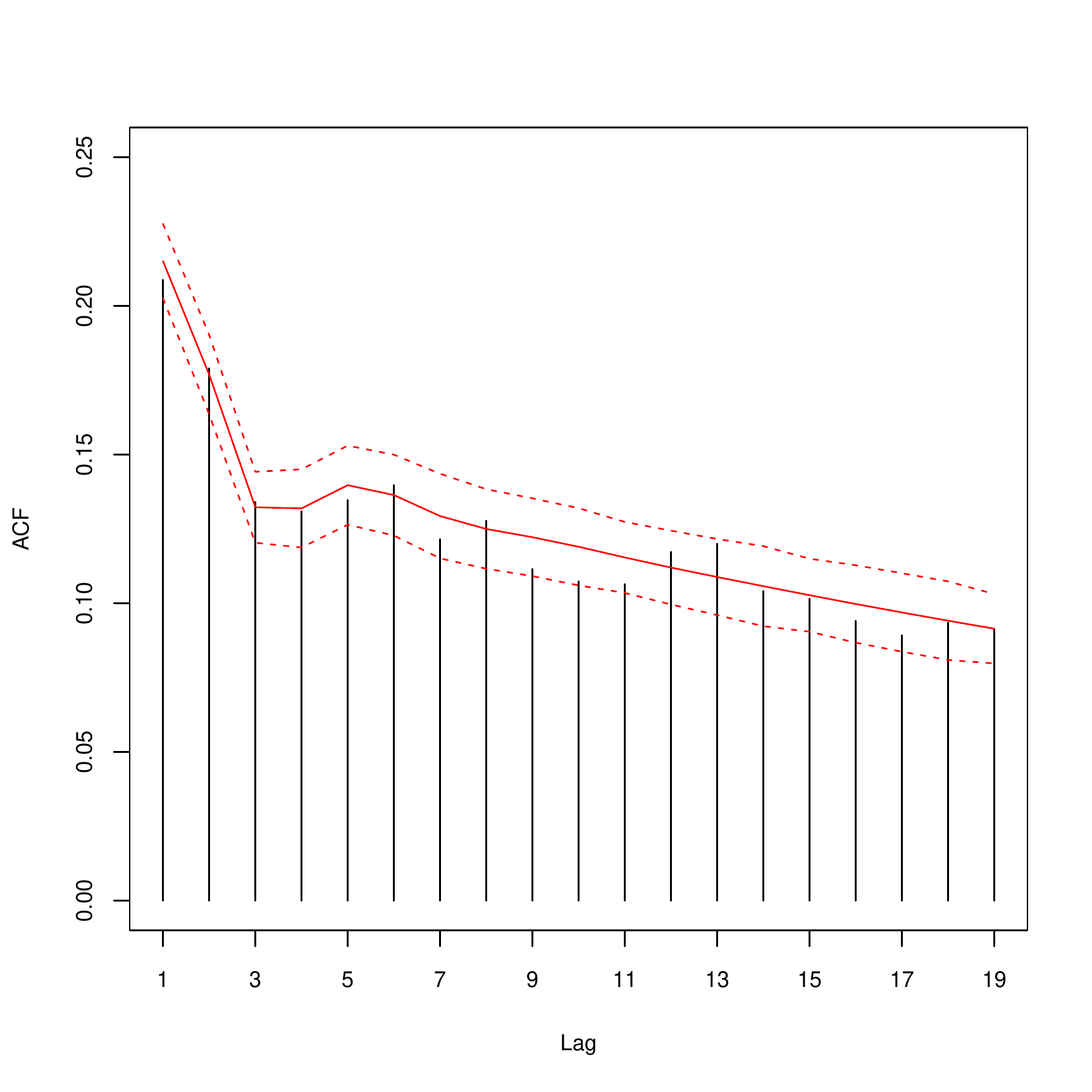}
   \caption{Autocorrelation function}
	 \label{AcfCARMA31}
\end{subfigure}
\caption{Trajectory of the counting process $N_t$ and autocorrelation function (ACF) using a CARMA(3,1)-Hawkes. Input parameters: $\mu =0.30$, $a_1 = 1.3$, $a_2 = 0.34+\pi^2/4 \approx 2.807 $, $a_3 = 0.025+0.025\pi^2 \approx 0.2717$, $b_0 = 0.2$ and $b_1 = 0.3$.}\label{fig:Carma31Hawkes}
\end{figure}

Table \ref{tab: MLE_CARMA31} shows the estimated parameters obtained using the MLE method and the number of occurred events for different levels of final time $T$, whereas Table \ref{tab: MME_CARMA31} exhibits the estimated parameters using the MME method for different levels of final time $T$.

\begin{table}[!h]
\centering
\begin{tabular}{@{}cccccccc@{}}
\toprule
$\hat{\mu}$ & $\hat{a}_1$ & $\hat{a}_2$ & $\hat{a}_3$  & $\hat{b}_0$ & $\hat{b}_1$ & $N_T$  & T  \\ \hline
0.3166 & 1.5372 & 2.8565 & 0.2869 & 0.2011 & 0.3345 &  5229  & 5000  \\
0.3063 & 1.3521 & 2.6686 & 0.2758 & 0.1998 & 0.3023 & 16666  & 15000 \\
0.3068 & 1.0811 & 2.4380 & 0.2480 & 0.1808 & 0.2508 & 28154  & 25000 \\
0.2949 & 1.4177 & 2.6901 & 0.2550 & 0.1889 & 0.3138 & 56815  & 50000 \\
\bottomrule
\end{tabular}
\caption{Parameter estimates from MLE and number of occurred events $N_T$ for a CARMA(3,1)-Hawkes process. True parameters are: $\mu =0.30$, $a_1 = 1.3$, $a_2 = 0.34+\pi^2/4 \approx 2.807 $, $a_3 = 0.025+0.025\pi^2 \approx 0.2717$, $b_0 = 0.2$ and $b_1 = 0.3$.}\label{tab: MLE_CARMA31}
\end{table}

\begin{table}[!h]
\centering
\begin{tabular}{@{}ccccccc@{}}
\toprule
$\hat{\mu}$ & $\hat{a}_1$ & $\hat{a}_2$ & $\hat{a}_3$  & $\hat{b}_0$ & $\hat{b}_1$ & T  \\ \hline
0.2946 & 2.4941 & 2.9148 & 0.2228 & 0.1600 & 0.4954  & 4999 \\
0.2788 & 2.8938 & 3.2043 & 0.2440 & 0.1827 & 0.6283  & 14999\\
0.2934 & 1.5754 & 2.8956 & 0.2634 & 0.1947 & 0.3592  & 24999\\
0.3104 & 1.2584 & 2.7107 & 0.2749 & 0.1998 & 0.2797  & 49999\\
\bottomrule
\end{tabular}
\caption{Parameter estimates from MME for a CARMA(3,1)-Hawkes process. True parameters are: $\mu =0.30$, $a_1 = 1.3$, $a_2 = 0.34+\pi^2/4 \approx 2.807 $, $a_3 = 0.025+0.025\pi^2 \approx 0.2717$, $b_0 = 0.2$ and $b_1 = 0.3$.}\label{tab: MME_CARMA31}
\end{table}




\section{Conclusion}
\label{concl}
In this paper we introduce a new Hawkes process where the intensity is a CARMA(p,q) model. We analyze the statistical properties of this process and obtain a closed-form expression for the autocorrelation function of the number of jumps observed in non-overlapping time intervals of same length. The model is a generalization of the standard Hawkes with exponential kernel but it is able to reproduce more complex dependence structures similar to those documented in several data sets. 

\section*{Acknowledgments}
This work was partly supported by JST CREST Grant Number JPMJCR14D7, Japan.

\bibliography{mybibfile}

\begin{appendix}
\section{Integration of matrix exponentials}
\label{Appendix2}
Let $\mathbf{A}$ be a square matrix and $\mathbf{A}^{\left(i\right)}:=\underset{i\text{ times}}{\underbrace{\mathbf{A}\mathbf{A} \cdots \mathbf{A}}}$. As the exponential of the matrix $\mathbf{A}$ can be computed as 
\begin{equation*}\exp\left(\mathbf{A}t\right)=\mathbf{I}+\overset{+\infty}{\underset{i=1}{\sum}}\frac{A^{\left(i\right)}t^i}{i!},
\end{equation*}
it is straightforward to show that
\begin{equation}
\int_{t_0}^{T}e^{\mathbf{A}\left(T-t\right)}\mbox{d}t=\mathbf{A}^{-1}\left(e^{\mathbf{A}\left(T-t_0\right)}-\mathbf{I}\right)=\left(e^{\mathbf{A}\left(T-t_0\right)}-\mathbf{I}\right)\mathbf{A}^{-1}.
\label{int2}
\end{equation}

\section{Solution of a general Linear Ordinary differential Equation}
To solve $\mbox{d}Y_t = \left(b_t+A Y_t\right) \mbox{d}t$, we consider the transformation $X_t=e^{-A t}Y_t$ and observe that
\begin{eqnarray*}
\mbox{d}X_t &=& -Ae^{-At}Y_t\mbox{d}t+e^{-At}\mbox{d}Y_t= e^{-At} b_t \mbox{d}t,
\end{eqnarray*}
from where we have $X_T = X_{t_0}+\int_{t_0}^T e^{-At} b_t \mbox{d}t$ that in terms of $Y_t$ reads 
\begin{equation}
Y_T=e^{A\left(T-t_0\right)}Y_{t_0}+\int_{t_0}^T e^{A\left(T-t\right)} b_t \mbox{d}t.
\label{solLinODE}
\end{equation}

\section{Computation of integrals with matrix exponentials}
\label{Appendix1}
Some useful results for computing integrals that involve matrix exponentials are provided in \cite{van1978computing} and  \cite{carbonell2008computing}. In particular, we recall the result that deals with the computation of the following two integrals:
\begin{equation}
\int_{0}^{t}e^{\mathbf{H}_{11}\left(t-u\right)} \mathbf{H}_{12} e^{\mathbf{H}_{22}u}\mbox{d}u
\label{B12}
\end{equation}
\begin{equation}
\int_{0}^{t}\int_{0}^{u}e^{\mathbf{H}_{11}\left(t-u\right)} \mathbf{H}_{12} e^{\mathbf{H}_{22}\left(u-r\right)} \mathbf{H}_{23} e^{\mathbf{H}_{33} r} \mbox{d}r\mbox{d}u
\label{B13}
\end{equation}
where $\mathbf{H}_{11}$, $\mathbf{H}_{12}$, $\mathbf{H}_{22}$, $\mathbf{H}_{23}$ and $\mathbf{H}_{33}$ have dimension $d_1\times d_1$,  $d_1\times d_2$, $d_2\times d_2$, $d_2\times d_3$ and  $d_3\times d_3$, respectively.
We  need to define   a block triangular matrix $\mathbf{H}$ as follows
\begin{equation}
\mathbf{H}:=\left(
\begin{array}{ccc}
\mathbf{H}_{11}&\mathbf{H}_{12} & \mathbf{0}\\
\mathbf{0} & \mathbf{H}_{22} & \mathbf{H}_{23}\\
\mathbf{0} & \mathbf{0} &  \mathbf{H}_{33}\\ 
\end{array}
\right).
\end{equation}
The integrals \eqref{B12} and \eqref{B13}  coincide with the elements $\mathbf{B}_{12}\left(t\right)$ and $\mathbf{B}_{13}\left(t\right)$  in the matrix exponential:
\begin{equation}
e^{\mathbf{H}t}=\left(\begin{array}{ccc}
\mathbf{B_{11}}\left(t\right) & \mathbf{B_{12}}\left(t\right) & \mathbf{B_{13}}\left(t\right)\\
\mathbf{0} & \mathbf{B_{22}}\left(t\right) &  \mathbf{B_{23}}\left(t\right)\\
\mathbf{0} & \mathbf{0} & \mathbf{B_{33}}\left(t\right)\\ 
\end{array}\right)
\label{A3}
\end{equation}
while $\mathbf{B}_{11}\left(t\right):=e^{\mathbf{H}_{11}t}$, $\mathbf{B}_{22}\left(t\right):=e^{\mathbf{H}_{22}t}$ and $\mathbf{B}_{33}\left(t\right):=e^{\mathbf{H}_{33}t}$.\newline
 
\begin{remark}The eigenvalues of $\mathbf{H}$ coincide with the eigenvalues of $\mathbf{H}_{11}$, $\mathbf{H}_{22}$ and $\mathbf{H}_{33}$. If the real part of all eigenvalues of $\mathbf{H}_{11}$, $\mathbf{H}_{22}$ and $\mathbf{H}_{33}$ is negative, the following result holds
\begin{equation*}
\lim_{t\rightarrow+\infty} e^{\mathbf{H}t}= \mathbf{0}
\end{equation*}
that implies
\begin{equation}
\lim_{t\rightarrow+\infty} \mathbf{B}_{12}\left(t\right)= \mathbf{0}
\label{veryImportant}
\end{equation}
and
\begin{equation}
\lim_{t\rightarrow+\infty} \mathbf{B}_{13}\left(t\right)= \mathbf{0}.
\label{veryImportant2}
\end{equation}
\end{remark}
\section{Proofs of propositions on long-run covariance and variance of the number of jumps}
\label{proofs}
We provide below the proof of Propositon \ref{PropAutocovCARMAHAWKES} on the long-run covariance of the number of jumps in a CARMA(p,q)-Hawkes model.

\begin{proof}
We first determine the covariance
of number of jumps in two non-overlapping time intervals given the information at time $t_0$. This quantity is formally defined as
\begin{eqnarray}
Cov_{t_0}\left(\tau,\delta\right)&:=& \mathbb{E}_{t_0}\left[\left(N_{t+\tau}-N_{t}\right)\left(N_{t+2\tau+\delta}-N_{t+\tau+\delta}\right)\right]\nonumber\\
&-&\mathbb{E}_{t_0}\left[\left(N_{t+\tau}-N_{t}\right)\right]\mathbb{E}_{t_0}\left[\left(N_{t+2\tau+\delta}-N_{t+\tau+\delta}\right)\right].
\label{Covt_0}
\end{eqnarray}
Using the iteration property of the conditional expected value, \eqref{Covt_0} becomes
\begin{eqnarray*}
Cov_{t_0}\left(\tau,\delta\right)&=&\mathbb{E}_{t_0}\left[\left(N_{t+\tau}-N_{t}\right)\mathbb{E}_{t+\tau}\left[\left(N_{t+2\tau+\delta}-N_{t+\tau+\delta}\right)\right]\right]\\
&-&\mathbb{E}_{t_0}\left[\left(N_{t+\tau}-N_{t}\right)\right]\mathbb{E}_{t_0}\left[\left(N_{t+2\tau+\delta}-N_{t+\tau+\delta}\right)\right].
\end{eqnarray*}
Applying the result \eqref{ExpNumbOfJump} in Proposition \ref{prop5}, we get
\begin{equation}
Cov_{t_0}\left(\tau,\delta\right)=\mathbf{b}^{\top}\tilde{\mathbf{A}}^{-1}\left[e^{\tilde{\mathbf{A}}\left(\tau+\delta\right)}-e^{\tilde{\mathbf{A}}\delta}\right] g_{t_0}\left(t,\tau\right)
\end{equation}
where
\begin{eqnarray}
g_{t_0}\left(t,\tau\right)&=&\mathbb{E}_{t_0}\left[\left(N_{t+\tau}-N_{t}\right)X_{t+\tau}\right]-e^{\tilde{\mathbf{A}}\left(t+\tau-t_0\right)}\mathbb{E}_{t_0}\left[N_{t+\tau}-N_{t}\right]X_{t_0}\nonumber\\
&+& \left(\mathbf{I}-e^{\tilde{\mathbf{A}}\left(t+\tau-t_0\right)}\right)\tilde{\mathbf{A}}^{-1}\mathbf{e}\mu\mathbb{E}_{t_0}\left[N_{t+\tau}-N_{t}\right]\nonumber\\
&=&\mathbb{E}_{t_0}\left[N_{t+\tau}X_{t+\tau}\right] +\tilde{\mathbf{A}}^{-1}\mathbf{e}\mu\mathbb{E}_{t_0}\left[N_{t}\right]
-e^{\tilde{\mathbf{A}}\tau}\left[\mathbb{E}_{t_0}\left(N_{t}X_{t}\right)+\tilde{\mathbf{A}}^{-1}\mathbf{e}\mu\mathbb{E}_{t_0}\left[N_{t}\right]\right]\nonumber\\
&-&e^{\tilde{\mathbf{A}}\left(t+\tau-t_0\right)}\mathbb{E}_{t_0}\left[N_{t+\tau}-N_{t}\right]X_{t_0}
+ \left(\mathbf{I}-e^{\tilde{\mathbf{A}}\left(t+\tau-t_0\right)}\right)\tilde{\mathbf{A}}^{-1}\mathbf{e}\mu\mathbb{E}_{t_0}\left[N_{t+\tau}-N_{t}\right].
\label{gt_0}
\end{eqnarray}
In the \textsl{rhs} of \eqref{gt_0}, the last two terms are stationary due to the result in \eqref{statExpNumbOfJump} and to the negativity of the real part for the eigenvalues of $\tilde{\mathbf{A}}$; the third term converges to zero as $t\rightarrow +\infty$ while the fourth term has the following limit behaviour
\begin{equation}
\left(\mathbf{I}-e^{\tilde{\mathbf{A}}\left(t+\tau-t_0\right)}\right)\tilde{\mathbf{A}}^{-1}\mathbf{e}\mu \mathbb{E}_{t_0}\left[N_{t+\tau}-N_{t}\right]\rightarrow \tilde{\mathbf{A}}^{-1}\mathbf{e}\mu^2\left(1-\mathbf{b}^\top \tilde{\mathbf{A}}^{-1}\mathbf{e}\right)\tau\text{ a.s. } t\rightarrow+\infty.
\label{limitfourthTermCov}
\end{equation}
For the first two terms in the \textsl{rhs} \eqref{gt_0} consider the quantity:
\begin{equation}
h_{t_0}\left(t,\tau\right):= \mathbb{E}_{t_0}\left[N_{t+\tau}X_{t+\tau}\right] +\tilde{\mathbf{A}}^{-1}\mathbf{e}\mu\mathbb{E}_{t_0}\left[N_{t}\right], \ \forall \tau\geq0, t>t_0
\label{hfunc}
\end{equation} as $t\rightarrow +\infty$. In \eqref{hfunc} the vector $\mathbb{E}_{t_0}\left[N_{t}X_{t}\right]$ requires the calculation of $p$ infinitesimal generators. We then observe that for $i=1,\ldots, p-1$, the infinitesimal generator of the function $N_t X_{t,i}$ is:
\begin{eqnarray*}
\mathcal{A}N_t X_{t,i}&=&\left(\mu+\mathbf{b}^\top X_{t}\right)\left[\left(N_t+1\right)X_{t,i}-N_tX_{t,i}\right] + N_{t}X_{t,i+1}\\
&=&\left(\mu X_{t,i}+X_{t,i} X_{t}^\top \mathbf{b} \right)+ N_{t}X_{t,i+1}
\end{eqnarray*}
while for $i=p$
\begin{eqnarray*}
\mathcal{A}N_t X_{t,p}&=&\left(\mu+\mathbf{b}^\top X_{t}\right)\left[\left(N_t+1\right)\left(X_{t,p}+1\right)-N_t X_{t,p}\right] + N_{t}A_{[p,\cdot]}X_{t}\\
&=&\left(\mu +\mathbf{b}^\top X_{t}+\mu N_{t}\right)+\left(\mu X_{t,p}+ X_{t,p}X_{t}^\top \mathbf{b}\right)+\left(\mathbf{b}^\top +  A_{[p,\cdot]}\right)N_{t}X_{t},
\end{eqnarray*}
that implies
\begin{equation}
\mbox{d}\mathbb{E}_{t_0}\left[X_{t}N_t\right]=\left[\left(\mu +\mathbf{b}^\top \mathbb{E}_{t_0}\left[X_{t}\right]+\mu \mathbb{E}_{t_0}\left[N_{t}\right]\right)\mathbf{e}+\mu\mathbb{E}_{t_0}\left[X_{t}\right]+\mathbb{E}_{t_0}\left[X_{t}X_{t}^{\top}\right]\mathbf{b}+\tilde{\mathbf{A}}\mathbb{E}_{t_0}\left[X_{t}N_t\right]\right]\mbox{d}t
\label{eq:ODE_NX}
\end{equation}
from where we get
\begin{eqnarray}
\mathbb{E}_{t_0}\left[X_{T}N_T\right]&=&e^{\tilde{\mathbf{A}}\left(T-t_0\right)}X_{t_0}N_{t_0}+\int_{t_0}^T e^{\tilde{\mathbf{A}}\left(T-t\right)}\left(\mu +\mathbf{b}^\top \mathbb{E}_{t_0}\left[X_{t}\right]+\mu \mathbb{E}_{t_0}\left[N_{t}\right]\right)\mathbf{e}\mbox{d}t\nonumber\\
&+& \int_{t_0}^T e^{\tilde{\mathbf{A}}\left(T-t\right)}\left[\mu\mathbb{E}_{t_0}\left[X_{t}\right]+\mathbb{E}_{t_0}\left[X_{t}X_{t}^{\top}\right]\mathbf{b}\right]\mbox{d}t.
\label{eq:InterExplSolXTNT}
\end{eqnarray}
The quantity $\mathbb{E}_{t_0}\left[X_{T}N_T\right]$ is not stationary but it is useful as it appears in the \textsl{rhs} of  the function $h_ {t_0}\left(t,\tau\right)$ introduced in \eqref{hfunc} that can be rewritten as
\begin{eqnarray}
h_{t_0}\left(t,\tau\right)&=&e^{\tilde{\mathbf{A}}\left(t+\tau-t_0\right)}X_{t_0}N_{t_0}+\int_{t_0}^{t+\tau} e^{\tilde{\mathbf{A}}\left(t+\tau-u\right)}\mu \mathbf{e}\mbox{d}u+\int_{t_0}^{t+\tau} e^{\tilde{\mathbf{A}}\left(t+\tau-u\right)}\mathbf{b}^\top \mathbb{E}_{t_0}\left[X_{u}\right]\mathbf{e}\mbox{d}u\nonumber\nonumber\\
&+&\int_{t_0}^{t+\tau} e^{\tilde{\mathbf{A}}\left(t+\tau-u\right)}\left(\mu \mathbb{E}_{t_0}\left[N_{u}\right]\right)\mathbf{e}\mbox{d}u+ \tilde{\mathbf{A}}^{-1} \mathbf{e}\mu\mathbb{E}_{t_0}\left[N_{t}\right] \nonumber \\
&+& \int_{t_0}^{t+\tau} e^{\tilde{\mathbf{A}}\left(t+\tau-u\right)}\left[\mu\mathbb{E}_{t_0}\left[X_{u}\right]+\mathbb{E}_{t_0}\left[X_{u}X_{u}^{\top}\right]\mathbf{b}\right]\mbox{d}u.
\label{myht0secondvers}
\end{eqnarray}
We analyze the long-run behaviour of each term in the \textsl{rhs} of \eqref{myht0secondvers}. We first observe that
\[
\int_{t_0}^{t+\tau} e^{\tilde{\mathbf{A}}\left(t+\tau-u\right)}\mbox{d}u\mu \mathbf{e}=\left(e^{\tilde{\mathbf{A}}\left(t+\tau-t_0\right)}-\mathbf{I}\right)\tilde{\mathbf{A}}^{-1}\mu \mathbf{e}
\] 
with
\begin{equation}
\lim_{t\rightarrow+\infty} \left(e^{\tilde{\mathbf{A}}\left(t+\tau-t_0\right)}-\mathbf{I}\right)\tilde{\mathbf{A}}^{-1}\mu \mathbf{e} = -\tilde{\mathbf{A}}^{-1}\mu \mathbf{e}.
\label{Limit2thterm}
\end{equation}
The formula for the conditional expected value of the state process in \eqref{conditionalStateProcMoM} allows us to rewrite the third term in the \textsl{rhs} of \eqref{myht0secondvers} as follows
\begin{eqnarray}
\int_{t_0}^{t+\tau} e^{\tilde{\mathbf{A}}\left(t+\tau-u\right)}\mathbf{e}\mathbf{b}^\top \mathbb{E}_{t_0}\left[X_{u}\right]\mbox{d}u &=&
e^{\tilde{\mathbf{A}}\left(t+\tau\right)}\int_{t_0}^{t+\tau} e^{-\tilde{\mathbf{A}}u}\mathbf{e}\mathbf{b}^\top e^{\tilde{\mathbf{A}}u}\mbox{d}u e^{-\tilde{\mathbf{A}}t_0}\left[X_{t_0}+\tilde{\mathbf{A}}^{-1}\mathbf{e}\mu\right]\nonumber\\
&-&\int_{t_0}^{t+\tau} e^{\tilde{\mathbf{A}}\left(t+\tau-u\right)}\mbox{d}u \mathbf{e}\mathbf{b}^\top \tilde{\mathbf{A}}^{-1}\mathbf{e}\mu \nonumber\\
&=& e^{\tilde{\mathbf{A}}\left(t+\tau\right)}\int_{t_0}^{t+\tau} e^{-\tilde{\mathbf{A}}u}\mathbf{e}\mathbf{b}^\top e^{\tilde{\mathbf{A}}u}\mbox{d}u e^{-\tilde{\mathbf{A}}t_0}\left[X_{t_0}+\tilde{\mathbf{A}}^{-1}\mathbf{e}\mu\right] \nonumber\\
&-& \left(e^{\tilde{\mathbf{A}}\left(t+\tau-t_0\right)}-\mathbf{I}\right) \tilde{\mathbf{A}}^{-1} \mathbf{e}\mathbf{b}^\top \tilde{\mathbf{A}}^{-1}\mathbf{e}\mu. 
\label{intermidiatethirdterm}
\end{eqnarray}
To compute the integral $e^{\tilde{\mathbf{A}}\left(t+\tau\right)}\int_{t_0}^{t+\tau} e^{-\tilde{\mathbf{A}}u}\mathbf{e}\mathbf{b}^\top e^{\tilde{\mathbf{A}}u}\mbox{d}u e^{-\tilde{\mathbf{A}}t_0}$ we use the result in \eqref{A3} and exploiting its limit behaviour \eqref{veryImportant}, the long-run behaviour  of \eqref{intermidiatethirdterm} becomes
\begin{equation}
\lim_{t\rightarrow+\infty} \int_{t_0}^{t+\tau} e^{\tilde{\mathbf{A}}\left(t+\tau-u\right)}\mathbf{e}\mathbf{b}^\top \mathbb{E}_{t_0}\left[X_{u}\right]\mbox{d}u = \tilde{\mathbf{A}}^{-1} \mathbf{e}\mathbf{b}^\top \tilde{\mathbf{A}}^{-1}\mathbf{e}\mu.
\label{Limit3thterm}
\end{equation}
The fourth term in the \textsl{rhs} of \eqref{myht0secondvers} can be written as
\begin{eqnarray}
& & \int_{t_0}^{t+\tau} e^{\tilde{\mathbf{A}}\left(t+\tau-u\right)}\mbox{d}u\mathbf{e} \mu N_{t_0}  + \int_{t_0}^{t+\tau} e^{\tilde{\mathbf{A}}\left(t+\tau-u\right)}\left(u-t_0\right) \mbox{d}u \mathbf{e}\mu^2\left(1-\mathbf{b}^{\top}\tilde{\mathbf{A}}^{-1}\mathbf{e}\right)\nonumber\\
&+&  \int_{t_0}^{t+\tau} e^{\tilde{\mathbf{A}}\left(t+\tau-u\right)}\mathbf{e}\mu\mathbf{b}^\top\tilde{\mathbf{A}}^{-1}
\left[e^{\tilde{\mathbf{A}}\left(u-t_0\right)}-\mathbf{I}\right]\mbox{d}u\left(X_{t_0}+\tilde{\mathbf{A}}^{-1}\mathbf{e}\mu\right)+ \tilde{\mathbf{A}}^{-1} \mathbf{e}\mu\mathbb{E}_{t_0}\left[N_{t}\right] \nonumber\\
&=&\left(e^{\tilde{\mathbf{A}}\left(t+\tau-t_0\right)}-\mathbf{I}\right)\tilde{\mathbf{A}}^{-1}\mathbf{e}\mu N_{t_0}+\left[\int_{t_0}^{t+\tau} e^{\tilde{\mathbf{A}}\left(t+\tau-u\right)}\left(u-t_0\right)\mbox{d}u\right]\mathbf{e}\mu^2\left(1-\mathbf{b}^{\top}\tilde{\mathbf{A}}^{-1}\mathbf{e}\right)\nonumber\\
&+& e^{\tilde{\mathbf{A}}\left(t+\tau\right)}\int_{t_0}^{t+\tau} e^{-\tilde{\mathbf{A}}u}\mathbf{e}\mathbf{b}^\top e^{\tilde{\mathbf{A}}u}\mbox{d}ue^{-\tilde{\mathbf{A}}t_0}\tilde{\mathbf{A}}^{-1}\left[X_{t_0}+\tilde{\mathbf{A}}^{-1}\mathbf{e}\mu\right]\mu\nonumber\\
&-& \tilde{\mathbf{A}}^{-1}\left(e^{\tilde{\mathbf{A}}\left(t+\tau-t_0\right)}-\mathbf{I}\right)\mathbf{e}\mathbf{b}^\top\tilde{\mathbf{A}}^{-1}\left[X_{t_0}+\tilde{\mathbf{A}}^{-1}\mathbf{e}\mu\right]\mu+\tilde{\mathbf{A}}^{-1} \mathbf{e}\mu\mathbb{E}_{t_0}\left[N_{t}\right]. 
\label{4thTerm}
\end{eqnarray}
Integrating by parts we get 
\begin{equation*}
\int_{t_0}^{t+\tau} e^{\tilde{\mathbf{A}}\left(t+\tau-u\right)}\left(u-t_0\right)\mbox{d}u =\tilde{\mathbf{A}}^{-1}\left[\left(e^{\tilde{\mathbf{A}}\left(t+\tau-t_0\right)}-\mathbf{I}\right)\tilde{\mathbf{A}}^{-1}-\mathbf{I}\left(t+\tau-t_0\right)\right].
\end{equation*} 
Thus \eqref{4thTerm} becomes
\begin{eqnarray*}
& & \left(e^{\tilde{\mathbf{A}}\left(t+\tau-t_0\right)}-\mathbf{I}\right)\tilde{\mathbf{A}}^{-1}\mathbf{e}\mu N_{t_0}+\tilde{\mathbf{A}}^{-1}\left[\left(e^{\tilde{\mathbf{A}}\left(t+\tau-t_0\right)}-\mathbf{I}\right)\tilde{\mathbf{A}}^{-1}-\mathbf{I}\left(t+\tau-t_0\right)\right]\mathbf{e}\mu^2\left(1-\mathbf{b}^{\top}\tilde{\mathbf{A}}^{-1}\mathbf{e}\right)\\
&+& e^{\tilde{\mathbf{A}}\left(t+\tau\right)}\int_{t_0}^{t+\tau} e^{-\tilde{\mathbf{A}}u}\mathbf{e}\mathbf{b}^\top e^{\tilde{\mathbf{A}}u}\mbox{d}ue^{-\tilde{\mathbf{A}}t_0}\tilde{\mathbf{A}}^{-1}\left[X_{t_0}+\tilde{\mathbf{A}}^{-1}\mathbf{e}\mu\right]\mu \\
&-& \tilde{\mathbf{A}}^{-1}\left(e^{\tilde{\mathbf{A}}\left(t+\tau-t_0\right)}-\mathbf{I}\right)\mathbf{e}\mathbf{b}^\top\tilde{\mathbf{A}}^{-1}\left[X_{t_0}+\tilde{\mathbf{A}}^{-1}\mathbf{e}\mu\right]\mu+ \tilde{\mathbf{A}}^{-1} \mathbf{e}\mu\mathbb{E}_{t_0}\left[N_{t}\right].
\end{eqnarray*}
Using the formula for the conditional expected value of the counting process in \eqref{condFirstMoMCount} we get
\begin{eqnarray*}
& &\left(e^{\tilde{\mathbf{A}}\left(t+\tau-t_0\right)}-\mathbf{I}\right)\tilde{\mathbf{A}}^{-1}\mathbf{e}\mu N_{t_0}+\tilde{\mathbf{A}}^{-1}\left[\left(e^{\tilde{\mathbf{A}}\left(t+\tau-t_0\right)}-\mathbf{I}\right)\tilde{\mathbf{A}}^{-1}-\mathbf{I}\left(t+\tau-t_0\right)\right]\mathbf{e}\mu^2\left(1-\mathbf{b}^{\top}\tilde{\mathbf{A}}^{-1}\mathbf{e}\right)\\
&+&e^{\tilde{\mathbf{A}}\left(t+\tau\right)}\int_{t_0}^{t+\tau} e^{-\tilde{\mathbf{A}}u}\mathbf{e}\mathbf{b}^\top e^{\tilde{\mathbf{A}}u}\mbox{d}ue^{-\tilde{\mathbf{A}}t_0}\tilde{\mathbf{A}}^{-1}\left[X_{t_0}+\tilde{\mathbf{A}}^{-1}\mathbf{e}\mu\right]\mu \\
&-& \tilde{\mathbf{A}}^{-1}\left(e^{\tilde{\mathbf{A}}\left(t+\tau-t_0\right)}-I\right)\mathbf{e}\mathbf{b}^\top\tilde{\mathbf{A}}^{-1}\left[X_{t_0}+\tilde{\mathbf{A}}^{-1}\mathbf{e}\mu\right]\mu\\
&+& \tilde{\mathbf{A}}^{-1} \mathbf{e}\mu\left[N_{t_0}+\mu\left(1-\mathbf{b}^{\top}\tilde{\mathbf{A}}^{-1}\mathbf{e}\right)\left(t-t_0\right)+\mathbf{b}^\top\tilde{\mathbf{A}}^{-1} \left(e^{\tilde{\mathbf{A}}\left(t_1-t_0\right)}-\mathbf{I}\right)\left[X_{t_0}+\tilde{\mathbf{A}}^{-1}\mathbf{e}\mu\right]\right] \\
&=& e^{\tilde{\mathbf{A}}\left(t+\tau-t_0\right)}\tilde{\mathbf{A}}^{-1}\mathbf{e}\mu N_{t_0}+\tilde{\mathbf{A}}^{-1}\left[\left(e^{\tilde{\mathbf{A}}\left(t+\tau-t_0\right)}-\mathbf{I}\right)\tilde{\mathbf{A}}^{-1}-\mathbf{I}\tau\right]\mathbf{e}\mu^2\left(1-\mathbf{b}^{\top}\tilde{\mathbf{A}}^{-1}\mathbf{e}\right)\\
&+&e^{\tilde{\mathbf{A}}\left(t+\tau\right)}\int_{t_0}^{t+\tau} e^{-\tilde{\mathbf{A}}u}\mathbf{e}\mathbf{b}^\top e^{\tilde{\mathbf{A}}u}\mbox{d}ue^{-\tilde{\mathbf{A}}t_0}\tilde{\mathbf{A}}^{-1}\left[X_{t_0}+\tilde{\mathbf{A}}^{-1}\mathbf{e}\mu\right]\mu \\
&-& \tilde{\mathbf{A}}^{-1}e^{\tilde{\mathbf{A}}\left(t+\tau-t_0\right)}\mathbf{e}\mathbf{b}^\top\tilde{\mathbf{A}}^{-1}\left[X_{t_0}+\tilde{\mathbf{A}}^{-1}\mathbf{e}\mu\right]\mu+ \tilde{\mathbf{A}}^{-1} \mathbf{e}\mu\mathbf{b}^\top\tilde{\mathbf{A}}^{-1} e^{\tilde{\mathbf{A}}\left(t-t_0\right)}\left[X_{t_0}+\tilde{\mathbf{A}}^{-1}\mathbf{e}\mu\right]
\end{eqnarray*}
and its long-run behaviour is established considering $t\rightarrow+\infty$, that is
\begin{equation}
-\tilde{\mathbf{A}}^{-1}\left[\tilde{\mathbf{A}}^{-1}+\mathbf{I}\tau\right]\mathbf{e}\mu^2\left(1-\mathbf{b}^{\top}\tilde{\mathbf{A}}^{-1}\mathbf{e}\right).
\label{Limit4thterm}
\end{equation}
The fifth term in the right-hand side of \eqref{myht0secondvers} can be rewritten as
\begin{eqnarray*}
\int_{t_0}^{t+\tau} e^{\tilde{\mathbf{A}}\left(t+\tau-u\right)} \mathbb{E}_{t_0}\left[X_{u}\right]\mbox{d}u \mu&=&\int_{t_0}^{t+\tau} e^{\tilde{\mathbf{A}}\left(t+\tau-u\right)} e^{\tilde{\mathbf{A}}\left(u-t_0\right)}\mbox{d}u \left[X_{t_0}+\tilde{\mathbf{A}}^{-1}\mathbf{e}\mu\right]\mu\\
&-& \int_{t_0}^{t+\tau} e^{\tilde{\mathbf{A}}\left(t+\tau-u\right)} \mbox{d}u \tilde{\mathbf{A}}^{-1}\mathbf{e}\mu^2\\ 
&=&  e^{\tilde{\mathbf{A}}\left(t+\tau-t_0\right)} \left(t+\tau-t_0\right)  \left[X_{t_0}+\tilde{\mathbf{A}}^{-1}\mathbf{e}\mu\right]\mu\\&-& \tilde{\mathbf{A}}^{-1}\left(e^{\tilde{\mathbf{A}}\left(t+\tau-t_0\right)}-\mathbf{I}\right)\tilde{\mathbf{A}}^{-1}\mathbf{e}\mu^2,
\end{eqnarray*}
that has the following long-run behaviour
\begin{equation}
\lim_{t\rightarrow+\infty}\int_{t_0}^{t+\tau} e^{\tilde{\mathbf{A}}\left(t+\tau-u\right)} \mathbb{E}_{t_0}\left[X_{u}\right]\mbox{d}u \mu= \tilde{\mathbf{A}}^{-1}\tilde{\mathbf{A}}^{-1}\mathbf{e}\mu^2.
\label{Limit5thterm}
\end{equation}

Lemma \ref{lemmavecttrill} suggests that the last term in the \textsl{rhs} of \eqref{myht0secondvers} can be written as
\begin{eqnarray*}
\int_{t_0}^{t+\tau} e^{\tilde{\mathbf{A}}\left(t+\tau-u\right)} \mathbb{E}_{t_0}\left[X_u X_u^{\top}\right]\mathbf{b}\mbox{d}u
&=& \int_{t_0}^{t+\tau} e^{\tilde{\mathbf{A}}\left(t+\tau-u\right)} \mathbf{B} e^{\tilde{\tilde{\mathbf{A}}}\left(u-t_0\right)}\mbox{d}uvlt\left(X_{t_0}X_{t_0}^{\top}\right)\\
&+& \int_{t_0}^{t+\tau} e^{\tilde{\mathbf{A}}\left(t+\tau-u\right)}\mathbf{B}e^{\tilde{\tilde{\mathbf{A}}}\left(u-t_0\right)}\mbox{d}u\tilde{\tilde{\mathbf{A}}}^{-1}\mu\left(\tilde{\mathbf{e}}-\tilde{\mathbf{C}}\tilde{\mathbf{A}}^{-1}\mathbf{e}\right) \nonumber\\
&-& \int_{t_0}^{t+\tau} e^{\tilde{\mathbf{A}}\left(t+\tau-u\right)}\mbox{d}u \mathbf{B}\tilde{\tilde{\mathbf{A}}}^{-1}\mu\left(\tilde{\mathbf{e}}-\tilde{\mathbf{C}}\tilde{\mathbf{A}}^{-1}\mathbf{e}\right)\nonumber\\
&+& \int_{t_0}^{t+\tau} e^{\tilde{\mathbf{A}}\left(t+\tau-u\right)} \mathbf{B} e^{\tilde{\tilde{\mathbf{A}}}u}\left[ \int_{t_0}^{u}e^{-\tilde{\tilde{\mathbf{A}}}h}\tilde{\mathbf{C}}e^{\tilde{\mathbf{A}}h}\mbox{d}h\right] e^{-\tilde{\mathbf{A}}t_0}\left[X_{t_0}+\tilde{\mathbf{A}}^{-1}\mathbf{e}\mu\right]\mbox{d}u.
\end{eqnarray*}
The result in \eqref{int2} implies that
\begin{eqnarray*}
\int_{t_0}^{t+\tau} e^{\tilde{\mathbf{A}}\left(t+\tau-u\right)} \mathbb{E}_{t_0}\left[X_u X_u^{\top}\right]\mathbf{b}\mbox{d}u
&=& \left[\int_{t_0}^{t+\tau} e^{\tilde{\mathbf{A}}\left(t+\tau-u\right)} \mathbf{B} e^{\tilde{\tilde{\mathbf{A}}}\left(u-t_0\right)}\mbox{d}u\right]vlt\left(X_{t_0}X_{t_0}^{\top}\right)\\
&+& \left[\int_{t_0}^{t+\tau} e^{\tilde{\mathbf{A}}\left(t+\tau-u\right)}\mathbf{B}e^{\tilde{\tilde{\mathbf{A}}}\left(u-t_0\right)}\mbox{d}u\right]\tilde{\tilde{\mathbf{A}}}^{-1}\mu\left(\tilde{\mathbf{e}}-\tilde{\mathbf{C}}\tilde{\mathbf{A}}^{-1}\mathbf{e}\right) \nonumber\\
&-& \left(e^{\tilde{\mathbf{A}}\left(t+\tau-t_0\right)}-\mathbf{I}\right)\tilde{\mathbf{A}}^{-1}\mathbf{B}\tilde{\tilde{\mathbf{A}}}^{-1}\mu\left(\tilde{\mathbf{e}}-\tilde{\mathbf{C}}\tilde{\mathbf{A}}^{-1}\mathbf{e}\right)\\
&+& \int_{t_0}^{t+\tau} e^{\tilde{\mathbf{A}}\left(t+\tau-u\right)} \mathbf{B} e^{\tilde{\tilde{\mathbf{A}}}u}\left[ \int_{t_0}^{u}e^{-\tilde{\tilde{\mathbf{A}}}h}\tilde{\mathbf{C}}e^{\tilde{\mathbf{A}}h}\mbox{d}h\right] e^{-\tilde{\mathbf{A}}t_0}\mbox{d}u\left[X_{t_0}+\tilde{\mathbf{A}}^{-1}\mathbf{e}\mu\right].
\end{eqnarray*}
To determine the asymptotic behaviour of this term, we analyze the long-run behaviour of the integral $\int_{t_0}^{t+\tau} e^{\tilde{\mathbf{A}}\left(t+\tau-u\right)} \mathbf{B} e^{\tilde{\tilde{\mathbf{A}}}\left(u-t_0\right)}\mbox{d}u$. Exploiting the result in \ref{Appendix1}, we have
\begin{equation*}
\lim_{t\rightarrow+\infty}\int_{t_0}^{t+\tau} e^{\tilde{\mathbf{A}}\left(t+\tau-u\right)} \mathbf{B} e^{\tilde{\tilde{\mathbf{A}}}\left(u-t_0\right)}\mbox{d}u =\mathbf{0}
\end{equation*}
as all eigenvalues of $\tilde{\mathbf{A}}$ and $\tilde{\tilde{\mathbf{A}}}$ have negative real part. Using the Fubini-Tonelli's Theorem the last integral becomes
\[
\int_{t_0}^{t+\tau} e^{\tilde{\mathbf{A}}\left(t+\tau-u\right)} \mathbf{B} e^{\tilde{\tilde{\mathbf{A}}}u}\left[ \int_{t_0}^{u}e^{-\tilde{\tilde{\mathbf{A}}}h}\tilde{\mathbf{C}}e^{\tilde{\mathbf{A}}h}\mbox{d}h\right] e^{-\tilde{\mathbf{A}}t_0}\mbox{d}u=\int_{t_0}^{t+\tau} \int_{t_0}^{u} e^{\tilde{\mathbf{A}}\left(t+\tau-u\right)} \mathbf{B} e^{\tilde{\tilde{\mathbf{A}}}\left(u-h\right)}\tilde{\mathbf{C}}e^{\tilde{\mathbf{A}}\left(h-t_0\right)}\mbox{d}h\mbox{d}u.
\]
Its long-run behaviour is obtained using the result in \eqref{veryImportant2}, that is
\begin{equation}
\lim_{t\rightarrow +\infty} \int_{t_0}^{t+\tau} \int_{t_0}^{u} e^{\tilde{\mathbf{A}}\left(t+\tau-u\right)} \mathbf{B} e^{\tilde{\tilde{\mathbf{A}}}\left(u-h\right)}\tilde{\mathbf{C}}e^{\tilde{\mathbf{A}}\left(h-t_0\right)}\mbox{d}h\mbox{d}u = \mathbf{0}.
\end{equation}
Finally, we have
\begin{equation}
\lim_{t\rightarrow +\infty} \int_{t_0}^{t+\tau} e^{\tilde{\mathbf{A}}\left(t+\tau-u\right)} \mathbb{E}_{t_0}\left[X_u X_u^{\top}\right]\mathbf{b}\mbox{d}u = \tilde{\mathbf{A}}^{-1}\mathbf{B}\tilde{\tilde{\mathbf{A}}}^{-1}\mu\left(\tilde{\mathbf{e}}-\tilde{\mathbf{C}}\tilde{\mathbf{A}}^{-1}\mathbf{e}\right).
\label{Limit6thterm}
\end{equation}
From \eqref{Limit2thterm}, \eqref{Limit3thterm}, \eqref{Limit4thterm}, \eqref{Limit5thterm} and \eqref{Limit6thterm} we obtain the limit behaviour for the quantity in \eqref{myht0secondvers}
\begin{eqnarray}
h_{\infty}\left(\tau\right)&:=&\lim_{t\rightarrow+\infty} h_{t_0}\left(t,\tau\right)\nonumber\\
&=& -\tilde{\mathbf{A}}^{-1}\mu \mathbf{e} +\tilde{\mathbf{A}}^{-1} \mathbf{e}\mathbf{b}^\top \tilde{\mathbf{A}}^{-1}\mathbf{e}\mu-\tilde{\mathbf{A}}^{-1}\left[\tilde{\mathbf{A}}^{-1}+\mathbf{I}\tau\right]\mathbf{e}\mu^2\left(1-\mathbf{b}^{\top}\tilde{\mathbf{A}}^{-1}\mathbf{e}\right)\nonumber\\
&+&\tilde{\mathbf{A}}^{-1}\tilde{\mathbf{A}}^{-1}\mathbf{e}\mu^2\tilde{\mathbf{A}}^{-1}\mathbf{B}\tilde{\tilde{\mathbf{A}}}^{-1}\mu\left(\tilde{\mathbf{e}}-\tilde{\mathbf{C}}\tilde{\mathbf{A}}^{-1}\mathbf{e}\right).
\label{Asyht0tau}
\end{eqnarray}
Using \eqref{Asyht0tau} we can determine the asymptotic behaviour of \eqref{gt_0} and we get
\begin{eqnarray}
g_{\infty}\left(\tau\right) &:=& \lim_{t\rightarrow +\infty} g_{t_0}\left(t,\tau\right)\nonumber\\
&=& \lim_{t\rightarrow +\infty} h_{t_0}\left(t,\tau\right)-e^{\tilde{\mathbf{A}}\tau}\left[\lim_{t\rightarrow +\infty} h_{t_0}\left(t,0\right)\right]+ \tilde{\mathbf{A}}^{-1}\mathbf{e}\mu^2\left(1-\mathbf{b}^\top \tilde{\mathbf{A}}^{-1}\mathbf{e}\right)\tau \nonumber\\ 
&=&  h_{\infty}\left(\tau\right)-e^{\tilde{\mathbf{A}}\tau}h_{\infty}\left(0\right)+\tilde{\mathbf{A}}^{-1}\mathbf{e}\mu^2\left(1-\mathbf{b}^\top \tilde{\mathbf{A}}^{-1}\mathbf{e}\right)\tau. 
\label{Interginfty}
\end{eqnarray}
By straightforward calculations \eqref{Interginfty} becomes \eqref{ginftytau} and the covariance reads as in \eqref{AutoCovCARMAHAWKES}.
\end{proof}

Here we provide the proof of Proposition \ref{PropVarCARMAHAWKES}.

\begin{proof}
For the asymptotic variance we need to compute the conditional variance of the number of jumps in an interval with length $\tau$. First we observe that
\begin{eqnarray*}
\sigma_{t_0}^2\left(t,\tau\right) &:=&\mathsf{Var}_{t_0}\left(N_{t+\tau}-N_{t}\right)=\mathbb{E}_{t_0}\left[\left(N_{t+\tau}-N_{t}\right)^2\right]-\mathbb{E}_{t_0}^2\left[N_{t+\tau}-N_{t}\right].
\end{eqnarray*}
We  then compute the second moment of the increments
\begin{eqnarray*}
\mathbb{E}_{t_0}\left[\left(N_{t+\tau}-N_{t}\right)^2\right]
&=& \mathbb{E}_{t_0}\left[N_{t+\tau}^2\right]+\mathbb{E}_{t_0}\left[N_{t}^2\right]-2\mathbb{E}_{t_0}\left[N_{t}\mathbb{E}_{t}\left[N_{t+\tau}\right]\right]\\
&=&\mathbb{E}_{t_0}\left[N_{t+\tau}^2\right]-\mathbb{E}_{t_0}\left[N_{t}^2\right]-2\mathbb{E}_{t_0}\left[N_{t}\right]\mu\left(1-\mathbf{b}^{\top}\tilde{\mathbf{A}}^{-1}\mathbf{e}\right)\tau\\
&-&2\mathbf{b}^{\top}\tilde{\mathbf{A}}^{-1}\left[e^{\tilde{\mathbf{A}}\tau}-\mathbf{I}\right]\left[\mathbb{E}_{t_0}\left[N_{t}X_{t}\right]+\tilde{\mathbf{A}}^{-1}\mathbf{e}\mu \mathbb{E}_{t_0}\left[N_{t}\right]\right].
\end{eqnarray*}
For $\mathbb{E}_{t_0}\left[N_{t}^2\right]$ it is useful to compute the infinitesimal operator for the function $f\left(X_{1,t},\ldots,X_{p,t},N_t,\right)=N_t^2$, that reads
\begin{equation*}
\mathcal{A}f_t=\mu\left(2 N_t+1\right)+2\mathbf{b}^{\top}N_t X_t + \mathbf{b}^{\top} X_t.
\end{equation*}
Applying the Dynkin's formula, we have
\begin{eqnarray*}
\mathbb{E}_{t_0}\left[N_{t}^2\right]&=&N_{t_0}^2+2\mu\int_{t_0}^{t}\mathbb{E}_{t_0}\left[N_u\right]\mbox{d}u+\mu\left(t-t_0\right)+2\mathbf{b}^{\top}\int_{t_0}^{t}\mathbb{E}_{t_0}\left[N_u X_u\right]\mbox{d}u + \mathbf{b}^{\top}\int_{t_0}^{t}\mathbb{E}_{t_0}\left[X_u\right]\mbox{d}t.
\end{eqnarray*}
Therefore
\begin{eqnarray}
\mathbb{E}_{t_0}\left[\left(N_{t+\tau}-N_{t}\right)^2\right]&=&2\mu\int_{t}^{t+\tau}\mathbb{E}_{t_0}\left[N_u\right]\mbox{d}u+\mu\tau+2\mathbf{b}^{\top}\int_{t}^{t+\tau}\mathbb{E}_{t_0}\left[N_u X_u\right]\mbox{d}u + \mathbf{b}^{\top}\int_{t}^{t+\tau}\mathbb{E}_{t_0}\left[X_u\right]\mbox{d}u\nonumber\\
&-&2\mathbb{E}_{t_0}\left[N_{t}\right]\mu\left(1-\mathbf{b}^{\top}\tilde{\mathbf{A}}^{-1}\mathbf{e}\right)\tau-2\mathbf{b}^{\top}\tilde{\mathbf{A}}^{-1}\left[e^{\tilde{\mathbf{A}}\tau}-\mathbf{I}\right]\left[\mathbb{E}_{t_0}\left[N_{t}X_{t}\right]+\tilde{\mathbf{A}}^{-1}\mathbf{e}\mu \mathbb{E}_{t_0}\left[N_{t}\right]\right]\nonumber\\
&=&2\mu\int_{t}^{t+\tau}\mathbb{E}_{t_0}\left[N_u-N_{t}\right]\mbox{d}u+\mu\tau+ \mathbf{b}^{\top}\int_{t}^{t+\tau}\mathbb{E}_{t_0}\left[X_u\right]\mbox{d}u\nonumber\\
&+&2\mathbf{b}^{\top}\int_{t}^{t+\tau}\left[\mathbb{E}_{t_0}\left[N_u X_u\right]+\tilde{\mathbf{A}}^{-1}\mathbf{e}\mu\mathbb{E}_{t_0}\left[N_{t}\right]\right]\mbox{d}u\nonumber \\
&-&2\mathbf{b}^{\top}\tilde{\mathbf{A}}^{-1}\left[e^{\tilde{\mathbf{A}}\tau}-\mathbf{I}\right]\left[\mathbb{E}_{t_0}\left[N_{t}X_{t}\right]+\tilde{\mathbf{A}}^{-1}\mathbf{e}\mu \mathbb{E}_{t_0}\left[N_{t}\right]\right].
\label{abcdef}
\end{eqnarray}
We study the asymptotic behaviour of the terms in \eqref{abcdef}.
We denote with $a_{t_0}\left(t,\tau\right):=\int_{t}^{t+\tau}\mathbb{E}_{t_0}\left[N_u-N_{t}\right]\mbox{d}u$
\begin{eqnarray*} 
a_{t_0}\left(t,\tau\right)&=& \int_{t}^{t+\tau} \mu \left(1-\mathbf{b}^{\top}\tilde{\mathbf{A}}^{-1}\mathbf{e}\right) \left(u-t\right)\mbox{d}u+ \int_{t}^{t+\tau} \mathbf{b}^{\top}\tilde{\mathbf{A}}^{-1}\left[e^{\tilde{\mathbf{A}}\left(u-t_0\right)}-e^{\tilde{\mathbf{A}}\left(t-t_0\right)}\right]\mbox{d}u\left[X_{t_0}+\tilde{\mathbf{A}}^{-1} \mathbf{e}\right]\\
&=& \mu \left(1-\mathbf{b}^{\top}\tilde{\mathbf{A}}^{-1}\mathbf{e}\right) \frac{\tau^2}{2}+  \int_{t}^{t+\tau} \mathbf{b}^{\top}\tilde{\mathbf{A}}^{-1}\left[e^{\tilde{\mathbf{A}}\left(u-t\right)}-\mathbf{I}\right]\mbox{d}u e^{\tilde{\mathbf{A}}\left(t-t_0\right)} \left[X_{t_0}+\tilde{\mathbf{A}}^{-1} \mathbf{e}\right].\\
\end{eqnarray*}
We observe that the following integral is finite
\[
\int_{t}^{t+\tau} \mathbf{b}^{\top}\tilde{\mathbf{A}}^{-1}\left[e^{\tilde{\mathbf{A}}\left(u-t\right)}-\mathbf{I}\right]\mbox{d}u<+\infty
\]
from where we deduce that
\begin{equation}
a_{\infty}\left(\tau\right):=\lim_{t\rightarrow +\infty} a_{t_0}\left(t,\tau\right)= \mu \left(1-\mathbf{b}^{\top}\tilde{\mathbf{A}}^{-1}\mathbf{e}\right) \frac{\tau^2}{2}.
\label{ainftytau}
\end{equation}
We then concentrate on the quantity $b_{t_0}\left(t,\tau\right):=\mu\tau+ \mathbf{b}^{\top}\int_{t}^{t+\tau}\mathbb{E}_{t_0}\left[X_u\right]\mbox{d}u$ that through straightforward computations can be written as
\begin{eqnarray*}
b_{t_0}\left(t,\tau\right)&=& \mu\tau+ \mathbf{b}^{\top}\int_{t}^{t+\tau}\left[e^{\tilde{\mathbf{A}}\left(u-t_0\right)}\left(X_{t_0}+\tilde{\mathbf{A}}^{-1}\mathbf{e}\mu\right)-\tilde{\mathbf{A}}^{-1}\mathbf{e}\mu\right]\mbox{d}u\\
&=&\left(1-\mathbf{b}^{\top}\tilde{\mathbf{A}}^{-1}\mathbf{e}\right)\mu\tau+\mathbf{b}^{\top}e^{\tilde{\mathbf{A}}\left(t-t_0\right)}\int_{t}^{t+\tau}e^{\tilde{\mathbf{A}}\left(u-t\right)}\left(X_{t_0}+\tilde{\mathbf{A}}^{-1}\mathbf{e}\mu\right)\mbox{d}u.
\end{eqnarray*}
Since we have a continuous integrand in a compact support
\[
\int_{t}^{t+\tau}e^{\tilde{\mathbf{A}}\left(u-t\right)}\mbox{d}u<+\infty,
\]
we have
\begin{equation}
b_{\infty}\left(\tau\right):=\lim_{t\rightarrow+\infty}b_{t_0}\left(t,\tau\right)= \left(1-\mathbf{b}^{\top}\tilde{\mathbf{A}}^{-1}\mathbf{e}\right)\mu\tau.
\label{binftytau}
\end{equation}
Denoting with $c_{t_0}\left(t,\tau\right):=\int_{t}^{t+\tau}\left[\mathbb{E}_{t_0}\left[N_u X_u\right]+\tilde{\mathbf{A}}^{-1}\mathbf{e}\mu\mathbb{E}_{t_0}\left[N_{t}\right]\right]\mbox{d}u$, we obtain
\begin{eqnarray*}
c_{t_0}\left(t,\tau\right)&=& I_{0,t_0}\left(t,\tau\right)+I_{1,t_0}\left(t,\tau\right)+I_{2,t_0}\left(t,\tau\right)+I_{3,t_0}\left(t,\tau\right)+I_{4,t_0}\left(t,\tau\right)+I_{5,t_0}\left(t,\tau\right)
\end{eqnarray*}
where $I_{0,t_0}\left(t,\tau\right):= \int_{t}^{t+\tau}e^{\tilde{\mathbf{A}}\left(u-t_0\right)}X_{t_0}N_{t_0}\mbox{d}u$ is rewritten as
\[
I_{0,t_0}\left(t,\tau\right)=e^{\tilde{\mathbf{A}(t-t_0)}\tau}\int_{t}^{t+\tau}e^{\tilde{\mathbf{A}}\left(u-t\right)}X_{t_0}N_{t_0}\mbox{d}u
\]
and using the same arguments as above, we get
\begin{equation*}
I_{0,\infty}\left(t,\tau\right):=\lim_{t\rightarrow +\infty}I_{0,t_0}\left(t,\tau\right)=\mathbf{0}.
\end{equation*}
The quantity $I_{1,t_0}\left(t,\tau\right):= \int_{t}^{t+\tau}\left(e^{\tilde{\mathbf{A}}\left(u-t_0\right)}
-\mathbf{I}\right)\tilde{\mathbf{A}}^{-1}\mathbf{e}\mu\mbox{d}u$ can be rewritten as
\[
I_{1,t_0}\left(t,\tau\right)= \int_{t}^{t+\tau}e^{\tilde{\mathbf{A}}\left(u-t_0\right)}
\tilde{\mathbf{A}}^{-1}\mathbf{e}\mu\mbox{d}u-\tilde{\mathbf{A}}^{-1}\mathbf{e}\mu\tau
\]
while taking the limit as $t \rightarrow +\infty$, we have
\begin{equation}
I_{1,\infty}\left(t,\tau\right):=\lim_{t\rightarrow +\infty}I_{1,t_0}\left(t,\tau\right)=-\tilde{\mathbf{A}}^{-1}\mathbf{e}\mu\tau.
\end{equation}
The quantity
\begin{eqnarray}
I_{2,t_0}\left(t,\tau\right)&:= &\int_{t}^{t+\tau}\int_{t_0}^{u}e^{\tilde{\mathbf{A}}\left(u-s\right)}\mathbf{e} \mathbf{b}^{\top}e^{\tilde{\mathbf{A}}\left(s-t_0\right)}\mbox{d}s\mbox{d}u\left[X_{t_0}+\tilde{\mathbf{A}}^{-1}\mathbf{e}\mu\right]\nonumber\\
&-&\int_{t}^{t+\tau}\left(e^{\tilde{\mathbf{A}}\left(u-t_0\right)}-\mathbf{I}\right)\mbox{d}u \tilde{\mathbf{A}}^{-1}\mathbf{e} \mathbf{b}^{\top}\tilde{\mathbf{A}}^{-1}\mathbf{e}\mu
\end{eqnarray}
depends on the integral $\int_{t}^{t+\tau}\int_{t_0}^{u}e^{\tilde{\mathbf{A}}\left(u-s\right)}\mathbf{e} \mathbf{b}^{\top}e^{\tilde{\mathbf{A}}\left(s-t_0\right)}\mbox{d}s\mbox{d}u$ where from the substitutions $s-t_0=h$ and $r=u-t$ we get
\begin{equation}
\int_0^{\tau} \int_0^{t+r-t_0} e^{\tilde{\mathbf{A}}\left(t-t_0+r-h\right)}\mathbf{e} \mathbf{b}^{\top}e^{\tilde{\mathbf{A}}h}\mbox{d}h\mbox{d}r.
\label{abcdefgh}
\end{equation}
Defining
\[
\ddot{\mathbf{A}}:=\left[
\begin{array}{cc}
\tilde{\mathbf{A}} & \mathbf{e}\mathbf{b}^{\top}\\
\mathbf{0}_{p,p} & \tilde{\mathbf{A}}
\end{array}
\right]
\]
and applying the result in \ref{Appendix1}, the inner integral in \eqref{abcdefgh} becomes
\begin{equation}
\left[\mathbf{I}_{p,p};\mathbf{0}_{p,p}\right]e^{\ddot{\mathbf{A}}(t-t_0+r)}\left[
\begin{array}{c}
\mathbf{0}_{p,p}\\
\mathbf{I}_{p,p}
\end{array}
\right].
\end{equation}
Thus the integral in \eqref{abcdefgh} can be computed as follows
\begin{equation}
\left[\mathbf{I}_{p,p};\mathbf{0}_{p,p}\right]e^{\ddot{\mathbf{A}}(t-t_0)}\int_0^{\tau}e^{\ddot{\mathbf{A}}r}\mbox{d}r\left[
\begin{array}{c}
\mathbf{0}_{p,p}\\
\mathbf{I}_{p,p}
\end{array}
\right].
\label{eqInt}
\end{equation}
We notice that as $\int_0^{\tau}e^{\ddot{\mathbf{A}}r}\mbox{d}r<+\infty$ and all eigenvalues of $\ddot{\mathbf{A}}$ have negative real part, then
\begin{equation*}
I_{2,\infty}\left(\tau\right):=\lim_{t\rightarrow +\infty} I_{2,t_0}\left(t,\tau\right)= \tilde{\mathbf{A}}^{-1}\mathbf{e} \mathbf{b}^{\top}\tilde{\mathbf{A}}^{-1}\mathbf{e}\mu \tau.
\end{equation*}
\item Similarly, we get the limit for the term $I_{3,t_0}\left(t,\tau\right):=\int_{t}^{t+\tau}\left[\int_{t_0}^{u}e^{\tilde{\mathbf{A}}\left(u-s\right)}\mu \mathbb{E}_{t_0}\left(N_s\right)\mathbf{e}\mbox{d}s+\tilde{\mathbf{A}}^{-1}\mathbf{e}\mu\mathbb{E}_{t_0}\left(N_u\right)\right]\mbox{d}u$  as $t\rightarrow +\infty$:
\begin{equation*}
I_{3,\infty}\left(t,\tau\right):=\lim_{t\rightarrow +\infty} I_{3,t_0}\left(t,\tau\right)=-\tilde{\mathbf{A}}^{-1}\left[\mathbf{I}\frac{\tau^2}{2}+\tilde{\mathbf{A}}^{-1}\tau\right]\mathbf{e}\mu^2\left(1-\mathbf{b}^{\top}\tilde{\mathbf{A}}^{-1}\mathbf{e}\right).
\end{equation*}
\item We define the following quantity
\begin{eqnarray*}
I_{4,t_0}\left(t,\tau\right)&:=& \left[\int_{t}^{t+\tau}e^{\tilde{\mathbf{A}}\left(u-t_0\right)}\left(u-t_0\right)\mbox{d}u\right]\left[X_{t_0}+\tilde{\mathbf{A}}^{-1}\mathbf{e}\mu\right]\mu+ \tilde{\mathbf{A}}^{-1}\tilde{\mathbf{A}}^{-1}\mathbf{e}\mu^2\tau\\
&-&\tilde{\mathbf{A}}^{-1}\int_{t}^{t+\tau}e^{\tilde{\mathbf{A}}\left(u-t_0\right)}\mbox{d}u \tilde{\mathbf{A}}^{-1}\mathbf{e}\mu^2
\end{eqnarray*}
and observe that the first integral can be rewritten as
\begin{eqnarray*}
\int_{t}^{t+\tau}e^{\tilde{\mathbf{A}}\left(u-t_0\right)}\left(u-t_0\right)\mbox{d}u&=& e^{\tilde{\mathbf{A}}\left(t-t_0\right)} \int_{t}^{t+\tau}e^{\tilde{\mathbf{A}}\left(u-t\right)}\left(u-t\right)\mbox{d}u+e^{\tilde{\mathbf{A}}\left(t-t_0\right)}\left(t-t_0\right)\int_{t}^{t+\tau}e^{\tilde{\mathbf{A}}\left(u-t\right)}\mbox{d}u\\
\end{eqnarray*}
where both terms in the \textsl{rhs} tend to zero as $t\rightarrow +\infty$ thus
\[
I_{4,\infty}\left(\tau\right)=\lim_{t\rightarrow +\infty}I_{4,t_0}\left(t,\tau\right)=\tilde{\mathbf{A}}^{-1}\tilde{\mathbf{A}}^{-1}\mathbf{e}\mu^2\tau.
\]
Similar arguments are used to determine the limit as $t \rightarrow +\infty$ for the quantity $I_{5,t_0}\left(\tau\right):=\int_{t}^{t+\tau}\int_{t_0}^{u}e^{\tilde{\mathbb{A}}\left(u-s\right)}\mathbb{E}_{t_0}\left[X_s,X_s^{\top}\right]\mathbf{b}\mbox{d}s\mbox{d}u$ as follows
\begin{equation*}
I_{5,\infty}\left(\tau\right)=\lim_{t\rightarrow+\infty}I_{5,t_0}\left(t,\tau\right)= \tilde{\mathbf{A}}^{-1}\mathbf{B}\tilde{\tilde{\mathbf{A}}}^{-1}\mu\left(\tilde{\mathbf{e}}-\tilde{\mathbf{C}}\tilde{\mathbf{A}}^{-1}\mathbf{e}\right)\tau.
\end{equation*}
Combining all results, we get the stationary behaviour for the quantity $c_{\infty}\left(\tau\right):=\lim_{t\rightarrow +\infty} c_{t_0}\left(t,\tau\right)$ that reads
\begin{eqnarray}
c_{\infty}\left(\tau\right)&=&-\tilde{\mathbf{A}}^{-1}\mathbf{e}\mu\tau\left(1-\mathbf{b}^{\top}\tilde{\mathbf{A}}^{-1}\mathbf{e}\right)-\tilde{\mathbf{A}}^{-1}\left[\mathbf{I}\frac{\tau^2}{2}+\tilde{\mathbf{A}}^{-1}\tau\right]\mathbf{e}\mu^2\left(1-\mathbf{b}^{\top}\tilde{\mathbf{A}}^{-1}\mathbf{e}\right)+\tilde{\mathbf{A}}^{-1}\tilde{\mathbf{A}}^{-1}\mathbf{e}\mu^2\tau\ \nonumber \\
&+&\tilde{\mathbf{A}}^{-1}\mathbf{B}\tilde{\tilde{\mathbf{A}}}^{-1}\mu\left(\tilde{\mathbf{e}}-\tilde{\mathbf{C}}\tilde{\mathbf{A}}^{-1}\mathbf{e}\right)\tau.
\label{cinftytau}
\end{eqnarray}
Furthermore
\begin{eqnarray*}
\lim_{t\rightarrow+\infty}\mathbb{E}_{t_0}\left[\left(N_{t+\tau}-N_{t}\right)^2\right]&=& 2\mu a_{\infty}\left(\tau\right)+ b_{\infty}\left(\tau\right)+2\mathbf{b}^{\top}c_{\infty}\left(\tau\right)-2\mathbf{b}^{\top}\tilde{\mathbf{A}}^{-1}\left[e^{\tilde{\mathbf{A}}\tau}-\mathbf{I}\right]h_{\infty}\left(0\right)\\
&=&\mu^2 \left(1-\mathbf{b}^{\top}\tilde{\mathbf{A}}^{-1}\mathbf{e}\right)^2\tau^2+\left(1-\mathbf{b}^{\top}\tilde{\mathbf{A}}^{-1}\mathbf{e}\right)\left(1-2\mathbf{b}^{\top}\tilde{\mathbf{A}}^{-1}\mathbf{e}\right)\mu\tau\\
&+&2\mathbf{b}^{\top}\tilde{\mathbf{A}}^{-1}\tilde{\mathbf{A}}^{-1}\mathbf{e}\tau\mu^2\left(\mathbf{b}^{\top}\tilde{\mathbf{A}}^{-1}\mathbf{e}\right)+2\mathbf{b}^{\top}\tilde{\mathbf{A}}^{-1}\mathbf{B}\tilde{\tilde{\mathbf{A}}}^{-1}\mu\left(\tilde{\mathbf{e}}-\tilde{\mathbf{C}}\tilde{\mathbf{A}}^{-1}\mathbf{e}\right)\tau\\
&-&2\mathbf{b}^{\top}\tilde{\mathbf{A}}^{-1}\left[e^{\tilde{\mathbf{A}}\tau}-\mathbf{I}\right]h_{\infty}\left(0\right).
\end{eqnarray*}
By  straightforward calculations, we obtain the  result in \eqref{AsymptoticVarianceJumps} for the  asymptotic  variance.
\end{proof}


\end{appendix}
\end{document}